\documentclass[draftcls,journal,onecolumn]{IEEEtran}

\usepackage{times}

\usepackage{amssymb}
\usepackage{amsthm}
\usepackage{amsxtra}
\usepackage{amsmath}
\usepackage{array}
\usepackage{algpseudocode}
\usepackage{algorithm}

\usepackage{verbatim}
\usepackage{epsfig}
\usepackage{setspace}
\usepackage{caption}
\usepackage{subcaption}
\usepackage{breqn}

\usepackage{hyperref}

\allowdisplaybreaks

\newtheorem{theorem}{Theorem}
\newtheorem{definition}{Definition}
\newtheorem{lemma}[theorem]{Lemma}
\newtheorem{proposition}[theorem]{Proposition}
\newtheorem{corollary}[theorem]{Corollary}
\newtheorem{example}{Example}[section]

\newtheorem{claim}{Claim}
\newtheorem*{remark*}{Remarks}

\def\CI{\text{CI}}

\def\CO{\text{CO}}
\def\SK{\text{SK}}

\def\cP{\mathcal{P}}
\def\cS{\mathcal{S}}
\def\cE{\mathcal{E}}
\def\cF{\mathcal{F}}
\def\cG{\mathcal{G}}
\def\cM{\mathcal{M}}
\def\cV{\mathcal{V}}
\def\cX{\mathcal{X}}
\def\BL{\textbf{L}}
\def\BK{\textbf{K}}
\def\BF{\textbf{F}}
\def\cH{\mathcal{H}}
\def\Ixm{\textbf{I}(X_{\mathcal{M}})}
\def\Itxm{\textbf{I}_T(X_{\mathcal{M}})}
\def\RTmin{R_{T}^{(\min)}}
\def\cR{\mathcal{R}}
\def\sfR{\mathsf{R}}
\def\cC{\mathcal{C}}
\def\STS{\text{STS}}
\def\Ixml{\textbf{I}(X_{\cM}^n|\textbf{L})}

\def\N{\mathbb{N}}

\def\beq{\begin{equation}}
\def\eeq{\end{equation}}

\title{On the Public Communication Needed \\ to Achieve SK Capacity \\ in the Multiterminal Source Model$^*$}

\author{
\IEEEauthorblockN{Manuj Mukherjee$^\dag$} \hspace{1cm} \and \IEEEauthorblockN{Navin Kashyap$^\dag$} \hspace {1cm} \and \IEEEauthorblockN{Yogesh Sankarasubramaniam$^\ddag$}
}

\begin{document}

\maketitle

\renewcommand{\thefootnote}{}
\footnotetext{
\noindent $^\dag$M.\ Mukherjee and N.\ Kashyap are with the 
Department of Electrical Communication Engineering, 
Indian Institute of Science, Bangalore. Email: \{manuj,nkashyap\}@ece.iisc.ernet.in.

\smallskip

$^\ddag$Email: yogesh@gatech.edu
}

\renewcommand{\thefootnote}{}
\footnotetext{$^*$This work was supported in part by a Swarnajayanti Fellowship granted by the Department of Science and Technology, India. Parts of this work were presented at the 2014 IEEE International Symposium on Information Theory (ISIT 2014), Honolulu, Hawaii, USA, and at ISIT 2015, Hong Kong, China.}

\renewcommand{\thefootnote}{\arabic{footnote}}

\begin{abstract}
The focus of this paper is on the public communication required for generating a maximal-rate secret key (SK) within the multiterminal source model of Csisz{\'a}r and Narayan. Building on the prior work of Tyagi for the two-terminal scenario, we derive a lower bound on the communication complexity, $R_{\SK}$, defined to be the minimum rate of public communication needed to generate a maximal-rate SK. It is well known that the minimum rate of communication for omniscience, denoted by $R_{\CO}$, is an upper bound on $R_{\SK}$. For the class of pairwise independent network (PIN) models defined on uniform hypergraphs, we show that a certain ``Type $\cS$'' condition, which is verifiable in polynomial time, guarantees that our lower bound on $R_{\SK}$ meets the $R_{\CO}$ upper bound. Thus, PIN models satisfying our condition are $R_{\SK}$-maximal, meaning that the upper bound $R_{\SK} \le R_{\CO}$ holds with equality. This allows us to explicitly evaluate $R_{\SK}$ for such PIN models.  We also give several examples of PIN models that satisfy our Type $\cS$ condition. 
Finally, we prove that for an arbitrary multiterminal source model, a stricter version of our Type $\cS$ condition implies that communication from \emph{all} terminals (``omnivocality'') is needed for establishing a SK of maximum rate. For three-terminal source models, the converse is also true: omnivocality is needed for generating a maximal-rate SK only if the strict Type $\cS$ condition is satisfied. Counterexamples exist that show that the converse is not true in general for source models with four or more terminals.
\end{abstract}

\section{Introduction}\label{sec:intro}

Maurer \cite{Maurer93} and Ahlswede and Csisz{\'a}r \cite{AC93} independently introduced the problem of generating a secret key (SK) for a pair of terminals observing distinct, albeit correlated, components of a discrete memoryless multi-component source. The SK is to be generated by communicating interactively over a noiseless public channel, and it is to be kept secure from all passive eavesdroppers having access to the public channel. The problem was subsequently extended to a multiterminal setting by Csisz{\'a}r and Narayan \cite{CN04}. The Csisz{\'a}r-Narayan model is now commonly referred to as the \emph{multiterminal source model}. The quantity of interest in these papers, and indeed in much of the literature that followed on this topic \cite{CN08}, \cite{NYBNR10}, \cite{NN10},  \cite{GA10}, is the \emph{secret key capacity}, i.e., the supremum of the rates of SK that can be generated within this model.  In the two-terminal case, an exact characterization of the SK capacity can be found in the original works of Maurer \cite{Maurer93} and Ahlswede and Csisz{\'a}r \cite{AC93}. Csisz{\'a}r and Narayan \cite{CN04} later gave an elegant single-letter expression for SK capacity in the general multiterminal source model. 

In all the aforementioned studies, the noiseless public channel is viewed as an unlimited free resource, and  no attempt is made to restrict the amount of communication sent through it. Indeed, Csisz{\'a}r and Narayan \cite[Section VI]{CN04} left open the question of determining the minimum rate of interactive public communication needed to achieve SK capacity. Tyagi \cite{Tyagi13} addressed this question in the two-terminal case, and gave an exact, although difficult to compute, characterization of the minimum rate of communication.  In this paper, we extend some of Tyagi's ideas to the multiterminal setting, and apply them to give an explicit answer to Csisz{\'a}r and Narayan's open question in some interesting special cases, namely, certain instances of the so-called pairwise independent network (PIN) model \cite{NYBNR10}, \cite{NN10}.

\subsection{Our Contributions}\label{sec:contrib}

The primary focus of our work is on the following question: \emph{What is the minimum rate of interactive public communication required to achieve SK capacity in the multiterminal source model?} We shall refer to the minimum rate of public communication as the \emph{communication complexity}\footnote{Our use of ``communication complexity'' differs from the use prevalent in the theoretical computer science literature where, following \cite{Yao79}, it refers to the total amount of communication, in bits, required to perform some distributed computation.} of achieving SK capacity, denoted by $R_{\SK}$. Csisz{\'a}r and Narayan's original proof of the achievability of their single-letter expression for SK capacity \cite[Theorem~1]{CN04} used a (non-interactive) communication protocol that enabled ``omniscience'' at all terminals, which means that the communication over the public channel allows each terminal to recover the observations of all the other terminals. It follows from their results that $R_{\SK}$ is always upper bounded by $R_{\CO}$, the minimum rate of interactive public communication required to achieve omniscience at all terminals. Furthermore, $R_{\CO}$ is given by the solution to a linear program \cite[Proposition~1]{CN04} (see \eqref{def:RCO} in Section~\ref{sec:prelim}), so it can be computed efficiently.  On the other hand, it is also well known that omniscience is not necessary for maximal-rate SK generation --- see the remark following Theorem~1 in \cite{CN04}, and also the proof of Theorem~3.2 in \cite{CN08}. Indeed, it is not difficult to find examples where $R_{SK} \ll R_{\CO}$; our Example~\ref{ex:omni} is one such. Thus, the sources for which we have $R_{\SK} = R_{\CO}$ constitute the worst-case sources in terms of communication complexity; we call such sources \emph{$R_{\SK}$-maximal}. We give a sufficient condition for a PIN model defined on a uniform hypergraph to be $R_{\SK}$-maximal, and show that PIN models satisfying this condition do exist. For these PIN models, it is easy to explicitly compute $R_{\CO}$, which then gives us an exact expression for $R_{\SK}$. This is the first (non-trivial) explicit evaluation of $R_{\SK}$ to be found in the literature, for a multiterminal source model with more than two terminals. Interestingly, for PIN models defined on ordinary graphs (i.e., each edge is incident with only two vertices), our sufficient condition is also necessary, which gives us an exact characterization, decidable in polynomial time, of ordinary graph PIN models that are $R_{\SK}$-maximal.

The main tool in our analysis is a lower bound on $R_{\SK}$ obtained via a multiterminal extension of Tyagi's work \cite{Tyagi13}. Tyagi's characterization of $R_{\SK}$ for two terminals \cite[Theorem~3]{Tyagi13} was in terms of the minimum rate of an \emph{interactive common information}, a type of Wyner common information (see \cite{Wyner75}). In order to appropriately generalize these ideas, we propose extensions of conditional mutual information and Wyner common information to the setting of more than two terminals. With these new multiterminal definitions in hand, we essentially follow the approach in \cite{Tyagi13} to derive a lower bound on $R_{\SK}$ in terms of the minimum rate of a multiterminal analogue of interactive common information. As in the case of Tyagi's result for two terminals, an exact evaluation of this bound appears to be a difficult task even for simple source models such as Markov chains. However, unlike the two-terminal result, we are unable to show that our lower bound to $R_{\SK}$ is tight in general. Luckily, we are able to evaluate this bound exactly for certain PIN models as mentioned above, and the bound turns out to be tight in these cases as it matches the $R_{\CO}$ upper bound. 

A secondary line of investigation carried out in this paper concerns the \emph{nature} of the public communication protocols that achieve SK capacity. It is well known that, in order to generate a maximal-rate SK in the two-terminal model, it is sufficient for only one terminal to communicate \cite{Maurer93, AC93, CN04}. All this terminal has to do is convey its local observations to the other terminal at the least possible rate of communication required to do so. Thus, in the two-terminal setup, it is \emph{never necessary} for both terminals to communicate to generate a capacity-achieving SK. Even in the case of more than two terminals, there are examples wherein not all terminals need to communicate --- again, see the remark following Theorem~1 in \cite{CN04}. However, as we will show in this paper, there are plenty of other examples where all terminals \emph{must} communicate in order to achieve SK capacity. We coin the term ``omnivocality'' to describe the state when all terminals communicate. The problem of interest to us then is the following: \emph{Characterize the instances of the multiterminal source model in which omnivocality is necessary for maximal-rate SK generation}. In this paper, we report some partial progress towards such a characterization. 

In \cite{GA10},  Gohari and Anantharam considered the scenario where a subset of terminals is required to remain silent, and yet all the terminals must agree upon an SK using only the communication from terminals that are allowed to talk. They derived a linear programming formulation for the maximum SK rate achievable in this scenario. Observe that omnivocality is necessary for achieving SK capacity in a source model iff any one terminal not being allowed to communicate strictly lowers the maximum achievable SK rate for that model. This establishes a correspondence between the omnivocality condition and the Gohari-Anantharam scenario involving silent terminals. We use this correspondence to identify a sufficient condition under which omnivocality is necessary for achieving SK capacity in a source model with at least three terminals. We further show that in the case of exactly three terminals, our sufficient condition is also necessary. Based on this evidence, we had conjectured in \cite{MKS14} that our condition was always necessary and sufficient. Unfortunately, a counterexample has been given by Chan et al.\ \cite{Chan14} that shows that the condition is not necessary for four or more terminals \cite{Chan14}. This has also been independently observed by Zhang et al.\ in \cite{ZLL15}. 

\subsection{Related work} \label{sec:rel}

Besides the work of Tyagi \cite{Tyagi13} that we have already mentioned, a few other recent papers have considered the SK generation problem from a communication complexity angle. The line of work that is perhaps most directly related to ours is that of Courtade and co-authors \cite{CW14,CHISIT14,CH14}, which considers the \emph{coded cooperative data exchange (CCDE)} problem \cite{ESS10} with the goal of generating an SK. This is, in essence, a single-shot version of the SK generation problem defined on hypergraph PIN models. Here, by ``single-shot'', we mean that each terminal sees only one realization of the component of the source available to it, as opposed to the Csisz{\'a}r-Narayan setup within which each terminal sees a sequence of  i.i.d.\ realizations. The single-shot SK capacity, i.e., the maximum size (as opposed to maximum rate) of an SK that can be generated, was evaluated in \cite[Theorem~6]{CW14}. The capacity achieving protocol used is a one-shot version of the communication-for-omniscience protocol of \cite{CN04}. The follow-up works \cite{CHISIT14} and \cite{CH14} addressed the issue of determining the minimum amount (again, as opposed to rate) of communication required to generate an SK of a particular size. However, this is done under a additional linearity requirement on the communication, i.e., the communication is required to be a linear function of the source outputs.  Theorem 11 of \cite{CH14} then gives an explicit characterization of what could be rightfully called the \emph{linear communication complexity} of generating an SK of a given size, in terms of the minimum number of hyperedges of an ``inherently $\tau$-connected subhypergraph''. It was further shown in \cite[Theorem 4]{CH14} that there exist hypergraph PIN models for which non-linear communication protocols for achieving (single-shot) SK capacity require lower amounts of communication than the linear communication complexity. It should be emphasized that our results do not make any linearity assumptions on the public communication.

In \cite{LCV15}, Liu et al.\ study public communication for SK generation in another variant of the multiterminal source model. The authors consider $m+1$ terminals observing correlated i.i.d. sources. One terminal acts as the communicator, sending information to each of the remaining $m$ terminals via $m$ different noiseless channels. A communication rate-key rate tradeoff region is identified for this model. However, the model is of somewhat limited interest to us because of the fact that each of the $m$ different links have individual eavesdroppers, but co-operation is not allowed among them. Secrecy is no longer guaranteed if the eavesdroppers co-operate. Therefore, the problem setup is more of an amalgam of two-terminal problems rather than a truly multiterminal setup. 

Communication complexity has also been studied for two-terminal interactive function computation without a secrecy constraint. Braverman and Rao in \cite[Theorem~II.3]{BR14} gave an exact characterization of the communication complexity for two-party interactive function computation.\footnote{To be precise, the quantity which we are calling communication complexity is referred to as \emph{amortized} communication complexity in \cite{BR14}.} The communication complexity is shown to be equal to an information-theoretic quantity called the \emph{internal information cost}. In a follow-up work, Braverman and Schneider provide an algorithm to compute the internal information cost for binary function computation --- see Theorem~1.1 of \cite{BS15}. 

Turning our attention to the topic of omnivocality originally studied in our paper \cite{MKS14}, Zhang et al.\ \cite{ZLL15} have recently obtained some new results. In particular, their Theorem~5 gives a sufficient condition for when a particular terminal \emph{must} communicate in any SK-capacity-achieving protocol. Our original sufficient condition for omnivocality \cite[Theorem~4]{MKS14}  (Theorem~\ref{th:mge3} in this paper) can now be obtained as a consequence of Zhang et al.'s Theorem~5. In addition, Theorem~4 of \cite{ZLL15} provides a sufficient condition that guarantees the existence of an SK-capacity-achieving protocol within which a given terminal can remain silent.

\subsection{Organization} \label{sec:org}
 
The paper is organized as follows. In Section~\ref{sec:prelim}, we provide the definitions and preliminaries needed for the rest of the paper. In Section~\ref{sec:commcomp}, we state and prove our lower bound on the communication complexity $R_{\SK}$. In Section~\ref{sec:RSKmax}, we identify a class of uniform hypergraph PIN models which are $R_{\SK}$-maximal. Section~\ref{sec:omnivocal} identifies a condition that makes omnivocality necessary for achieving SK capacity. The issue of verifying whether that condition holds for a given multiterminal source model is addressed in Section~\ref{sec:singleton}. Finally, Section~\ref{sec:conc} summarizes our results and presents some open problems. To preserve the flow of the exposition, the proofs of some of our results have been moved to appendices.

\section{Preliminaries} \label{sec:prelim}
 We start by giving a mathematical description of the multiterminal source model of \cite{CN04}. Throughout, we use $\N$ to denote the set of positive integers. Consider a set of $m$ terminals denoted by $\mathcal{M}=\{1,2, \ldots, m\}$. Each terminal $i \in \mathcal{M}$ observes $n$ i.i.d.\ repetitions of the random variable $X_i$ taking values in the finite set $\mathcal{X}_i$. The $n$ i.i.d.\ copies of the random variable are denoted by $X_i^n$. For any subset $A\subseteq \mathcal{M}$, $X_A$ and $X_A^n$ denote the collections of random variables $(X_i:i \in A)$ and $(X_i^n: i \in A)$, respectively. The terminals communicate through a noiseless public channel, any communication sent through which is accessible to all terminals and to potential eavesdroppers as well.
An \emph{interactive communication} is a communication $\textbf{f}=(f_1,f_2,\cdots,f_r)$ with finitely many transmissions $f_j$, in which any transmission sent by the $i$th terminal is a deterministic function of $X_i^n$ and all the previous communication, i.e., if terminal $i$ transmits $f_j$, then $f_j$ is a function only of $X_i^n$ and $f_1,\ldots,f_{j-1}$.   We denote the random variable associated with $\textbf{f}$ by $\textbf{F}$; the support of $\textbf{F}$ is a finite set $\cF$. The rate of the communication $\textbf{F}$ is defined as $\frac{1}{n}\log|\cF|$. Note that $\textbf{f}$, $\textbf{F}$ and $\cF$ implicitly depend on $n$.

\begin{definition}
\label{def:CR}
A \emph{common randomness (CR)} obtained from an interactive communication $\textbf{F}$ is a sequence of random variables $\textbf{J}^{(n)}$, $n\in\N$, which are functions of $X_{\mathcal{M}}^n$, such that for any $0<\epsilon<1$ and for all sufficiently large $n$, there exist $J_i=J_i(X_i^n,\textbf{F})$, $i = 1,2,\ldots,m$, satisfying $\text{Pr}[J_1=J_2=\cdots=J_m=\textbf{J}^{(n)}] \geq 1-\epsilon$.
\end{definition}

When $\textbf{J}^{(n)}=X_{\cM}^n$ we say that the terminals in $\cM$ have attained \emph{omniscience}. The communication $\textbf{F}$ which achieves this is called a \emph{communication for omniscience}. It was shown in Proposition~1 of \cite{CN04} that the minimum rate achievable by a communication for omniscience, denoted by $R_{\CO}$, is equal to $\displaystyle\min_{(R_1,R_2,\ldots,R_m)\in\mathcal{R}_{\CO}}\sum_{i=1}^mR_i$, where the region $\mathcal{R}_{\CO}$ is given by
\begin{equation}
\mathcal{R}_{\CO}=\biggl\{(R_1,R_2,\ldots,R_m): \sum_{i\in B}R_i\geq H(X_B|X_{B^c}), B\subsetneq\cM\biggr\}.
\label{def:RCO}
\end{equation}
Henceforth, we will refer to $R_{\CO}$ as the ``minimum rate of communication for omniscience''. Further, it can be seen from the description of $\mathcal{R}_{\CO}$ that $R_{\CO}<\infty$. More precisely, note that the point $(R_1,R_2,\ldots,R_m)$ defined by $R_i=H(X_i)$ for all $i$ lies in $\mathcal{R}_{\CO}$, and hence $R_{\CO}\leq\sum_{i=1}^mH(X_i)<\infty$.

\begin{definition}
\label{def:SK}
A real number $R\geq 0$ is an \emph{achievable SK rate} if there exists a CR $\textbf{K}^{(n)}$, $n \in \N$, obtained from an interactive communication $\textbf{F}$ satisfying, for any $\epsilon > 0$ and for all sufficiently large $n$, $I(\textbf{K}^{(n)};\textbf{F})\leq \epsilon$ and $\frac{1}{n}H(\textbf{K}^{(n)}) \geq R-\epsilon$. The \emph{SK capacity} is defined to be the supremum among all achievable rates.  The CR $\textbf{K}^{(n)}$ is called a \emph{secret key (SK)}. 
\end{definition}

From now on, we will drop the superscript $(n)$ from both $\textbf{J}^{(n)}$ and $\textbf{K}^{(n)}$ to keep the notation simple. 

The SK capacity can be expressed as \cite[Theorem~1]{CN04}
\begin{equation}
\cC(\cM)=H(X_{\cM})-R_{\CO}.\label{omni}
\end{equation}
Other equivalent characterizations of $\cC(\cM)$ exist in the literature. Csisz{\'a}r and Narayan observed in \cite[Section~V]{CN04} that since the linear program in \eqref{def:RCO} has an optimal solution, by strong duality, the dual linear program also has the same optimal value. Using this fact, the expression for SK capacity can be rewritten as
\begin{equation}
\cC(\cM) \triangleq H(X_{\mathcal{M}})-\max_{\lambda \in \Lambda} \sum_{B \in \mathcal{B}} \lambda_B H(X_B| X_{B^c}) \label{skcapacity}
\end{equation}
where $\mathcal{B}$ is the set of all non-empty, proper subsets of $\mathcal{M}$, and $\Lambda$ is the set of all \emph{fractional partitions} defined on $\mathcal{B}$. To be precise, any $\lambda=(\lambda_B: B\in \mathcal{B})\in \Lambda$ satisfies $\lambda_B\geq 0$, for all $B\in \mathcal{B}$, and $\sum_{B:i\in B}\lambda_B=1$, for all $i\in \mathcal{M}$. It is a fact that $H(X_{\mathcal{M}})-\max_{\lambda \in \Lambda} \sum_{B \in \mathcal{B}} \lambda_B H(X_B| X_{B^c}) \ge 0$ \cite[Proposition~II]{MT10}.
 
Another characterization of SK capacity can be given via the notion of \emph{multipartite information} defined as follows:
\begin{equation}
\Ixm\triangleq\min_{\cP}\Delta(\cP)
\label{eq:I}
\end{equation}
with $\Delta(\cP)\triangleq \frac{1}{|\cP|-1} \left[\sum_{A \in \cP} H(X_A) - H(X_{\cM}) \right]$ and the minimum being taken over all partitions $\cP=\{A_1,A_2,\cdots,A_{\ell}\}$ of $\cM$, of size $\ell \ge 2$. Note that $\mathbf{I}(X_{\cM}^n)=n\Ixm$. The quantity $\Ixm$ is a generalization of the mutual information to a multiterminal setting; indeed, for $m=2$, we have $\textbf{I}(X_1,X_2)=I(X_1;X_2)$. Note that $\mathbf{I}(X_{\cM}) = 0$ iff there exists a partition $\cP = \{A_1,A_2,\ldots,A_\ell\}$ of $\cM$, with $\ell \ge 2$, such that the random variables $X_{A_1}, X_{A_2},\ldots,X_{A_\ell}$ are mutually independent. It was shown in Theorem~1.1 of \cite{CZ10} and Theorem~4.1 of \cite{Chan14} that
\begin{equation}
\cC(\cM)=\Ixm.  \label{capeqmi}
\end{equation}
For the rest of this paper we shall use $\cC(\cM)$ and $\Ixm$ interchangeably. 

The partition $\bigl\{\{1\},\{2\},\ldots,\{m\}\bigr\}$ consisting of $m$ singleton cells will play a special role in the later sections of this paper; we call this the \emph{singleton partition} and denote it by $\cS$. The sources where \emph{$\cS$ is a minimizer for \eqref{eq:I}} will henceforth be referred to as \emph{Type $\cS$ sources}. If \emph{$\cS$ is the unique minimizer for \eqref{eq:I}} then we call such a source \emph{strict Type $\cS$}. A connection between the optimal fractional partition in \eqref{skcapacity} and the optimal partition in \eqref{eq:I} was pointed out in \cite{CZ10}. For any partition $\cP$ of $\cM$ define $\lambda^{(\cP)}$ as follows: $\lambda^{(\cP)}_B\triangleq\frac{\mathbb{I}\{B^c\in\cP\}}{|\cP|-1}$, for all $B\in\mathcal{B}$ and $\mathbb{I}\{.\}$ is the indicator function. It is easy to check that $\lambda^{(\cP)}$ is a fractional partition on $\mathcal{B}$, and $\Delta(\cP)=H(X_{\cM})-\sum_{B\in\mathcal{B}}\lambda^{(\cP)}_BH(X_B|X_{B^c})$. Hence, for any partition $\cP^*$ which is a minimizer in \eqref{eq:I}, the corresponding $\lambda^{(\cP^*)}$ is an optimal fractional partition for \eqref{skcapacity}.

We are now in a position to make the notion of communication complexity rigorous.

\begin{definition}
\label{def:RSKr}
A real number $R\geq 0$ is said to be an \emph{achievable rate of interactive communication for maximal-rate SK} if for all $\epsilon > 0$ and for all sufficiently large $n$, there exist \emph{(i)}~an interactive communication $\textbf{F}$ satisfying $\frac{1}{n}\log|\cF| \; \leq R+\epsilon$, and \emph{(ii)}~an SK $\textbf{K}$ obtained from $\textbf{F}$ such that $\frac{1}{n}H(\textbf{K})\geq \textbf{I}(X_{\mathcal{M}})-\epsilon$.

The infimum among all such achievable rates is called the \emph{communication complexity of achieving SK capacity}, denoted by $R_{\SK}$.
\end{definition}

The proof of Theorem~1 in \cite{CN04} shows that there exists an interactive communication $\textbf{F}$ that enables omniscience at all terminals and from which a maximal-rate SK can be obtained. Therefore, we have $R_{\SK}\leq R_{\CO}< \infty$. Hence, in terms of communication complexity, the sources that satisfy $R_{\SK}=R_{\CO}$ are the worst-case sources. We will henceforth refer to them as $R_{\SK}$-\emph{maximal} sources. Such sources do exist, as will be shown in Sections~\ref{sec:RSKmax} and \ref{sec:singleton}.

Tyagi gave a characterization of $R_{\SK}$ in the case of a two-terminal model \cite[Theorem~3]{Tyagi13}.\footnote{It should be clarified that Tyagi's characterization works only for ``weak'' SKs, which are defined as in our Definition~\ref{def:SK}, except that the condition $I(\textbf{K};\textbf{F}) \le \epsilon$ is weakened to $\frac{1}{n} I(\textbf{K};\textbf{F}) \le \epsilon$.  Using our definitions, Tyagi's arguments would only yield a two-terminal analogue of our Theorem~\ref{th:commcomp}.\label{fn:weakSK}} The key to his characterization was the observation that conditioned on a maximal-rate SK $\textbf{K}$ and the communication $\textbf{F}$ from which $\textbf{K}$ is extracted, the observations of the two terminals are ``almost'' independent: $\frac{1}{n}I(X_1^n;X_2^n | \textbf{K},\textbf{F}) \to 0$ as $n \to \infty$. Thus, the pair $(\textbf{K},\textbf{F})$ is a Wyner common information \cite{Wyner75} for the randomness at the terminals. Tyagi used the term ``interactive common information'' to denote any Wyner common information that consisted of a CR along with the interactive communication achieving it. 

We extend Tyagi's ideas to the setting of $m\geq 2$ terminals. We first extend the definition of conditional mutual information to the multiterminal setting. We will refer to the multiterminal analogue of the conditional mutual information as the \emph{conditional multipartite information}, and denote the conditional multipartite information of $X_{\cM}$ given a random variable $\BL$ by $\textbf{I}(X_{\cM}|\BL)$. As a natural extension of \eqref{eq:I}, we could define $\textbf{I}(X_{\cM}|\BL)$ as $\min_{\cP}\Delta(\cP|\BL)$, where $\Delta(\cP|\BL)\triangleq\frac{1}{|\cP|-1}\biggl[\sum_{A\in\cP}H(X_A|\BL)-H(X_{\cM}|\BL)\biggr]$. Note that 
\begin{equation}
\Delta(\cP|\BL)=H(X_{\cM}|\BL)-\sum_{B\in\mathcal{B}}\lambda^{(\cP)}_BH(X_B|X_{B^c},\BL). \label{lambdap}
\end{equation}
Using this definition, we can indeed generalize Tyagi's arguments to the case of $m\geq 2$ terminals and obtain a lower bound on $R_{\SK}$. It turns out, however, that a stronger lower bound can be obtained by defining $\mathbf{I}(X_{\cM}|\mathbf{L})$ to be equal to $\Delta(\cP^*|\BL)$, where $\cP^*$ is any partition that achieves the minimum in \eqref{eq:I}. One complication now is that there could be more than one choice of $\cP^*$ that achieves the minimum in \eqref{eq:I}, and two distinct choices of $\cP^*$ could yield different values for $\Delta(\cP^*|\BL)$. We simply choose the one that results in the largest value for $\mathbf{I}(X_{\cM}|\mathbf{L})$. Thus, we define 
\begin{equation} 
\textbf{I}(X_{\cM}|\BL)\triangleq\displaystyle\max_{\cP^*\in\text{argmin}_{\cP}\Delta(\cP)}\Delta(\cP^*|\BL). \label{cmi}
\end{equation}
The definition of $\textbf{I}(X_{\cM}|\BL)$ applies to any collection of jointly distributed random variables $X_{\cM}$; in particular it applies to the collection $X_{\cM}^n$. To be clear, 
$$
\Ixml\triangleq\displaystyle\max_{\cP^*\in\text{argmin}\Delta(\cP)}\frac{1}{|\cP^*|-1}\biggl[\sum_{A\in\cP^*}H(X_A^n|\BL)-H(X_{\cM}^n|\BL)\biggr].
$$

We point out an important consequence of our definition of $\textbf{I}(X_{\cM}|\BL)$. We have $\textbf{I}(X_{\cM}|\BL)=0$ iff for any partition $\cP^*$ which achieves the minimum in \eqref{eq:I}, the random variables $X_{\cM}$ are conditionally independent across the cells of $\cP^*$ given $\BL$. We are now in a position to extend the notion of the Wyner common information to a multipartite setting.\footnote{In fact, there exist other ways of generalizing Wyner common information to the multiterminal setting (see \cite{XLC13} and \cite{TSP11}).} 

\begin{definition}
\label{def:CI}
A \emph{(multiterminal) Wyner common information ($\CI_W$)} for $X_{\mathcal{M}}$ is a sequence of finite-valued functions $\textbf{L}^{(n)}=\textbf{L}^{(n)}(X_{\mathcal{M}}^n)$ such that $\frac{1}{n}\textbf{I}(X_{\mathcal{M}}^n|\textbf{L}^{(n)}) \to 0$ as $n \to \infty$. An \emph{interactive common information ($\CI$)} for $X_{\mathcal{M}}$ is a Wyner common information of the form $\textbf{L}^{(n)} = (\textbf{J},\textbf{F})$, where $\textbf{F}$ is an interactive communication and $\textbf{J}$ is a CR obtained from $\textbf{F}$. 
\end{definition}

Again, we shall drop the superscript $(n)$ from $\textbf{L}^{(n)}$ for notational simplicity. Wyner common informations $\textbf{L}$ do exist: for example, the identity map $\textbf{L}=X_{\mathcal{M}}^n$ is a $\CI_W$. To see that $\CI$s $(\textbf{J},\textbf{F})$ also exist, observe that $\textbf{J}=X_{\mathcal{M}}^n$ and a communication $\textbf{F}$ enabling omniscience constitute a $\CI_W$, and hence, a $\CI$.

\begin{definition}
\label{def:CIrate}
A real number $R\geq 0$ is an \emph{achievable $\CI_W$ (resp.\ $\CI$) rate} if there exists a $\CI_W$ $\textbf{L}$ (resp.\ a $\CI$ $\textbf{L} = (\textbf{J},\textbf{F})$) such that for all $\epsilon > 0$, we have
$\frac{1}{n}H(\textbf{L})\leq R+\epsilon$ for all sufficiently large $n$. We denote the infimum among all achievable $\CI_W$ (resp.\ $\CI$) rates by $\CI_W(X_{\mathcal{M}})$ (resp. \ $\CI(X_{\mathcal{M}})$).
\end{definition}

The proposition below records the relationships between some of the information-theoretic quantities defined so far. 

\begin{proposition} For a multiterminal source $X_{\cM}^n$, we have 
$H(X_{\mathcal{M}}) \ge \CI(X_{\mathcal{M}})\geq \CI_W(X_{\mathcal{M}})\geq \textbf{I}(X_{\mathcal{M}})$.
\label{prop:ineqs}
\end{proposition}
\begin{IEEEproof}
The first inequality is due to the fact that there exists a $\CI$ of rate $H(X_{\mathcal{M}})$. The second follows from the fact that a $\CI$ is a special type of $\CI_W$, so that $\CI(X_{\mathcal{M}}) \ge \CI_W(X_{\mathcal{M}})$.

For the last inequality, we start by observing that for any function $\textbf{L}$ of $X_{\mathcal{M}}^n$ and any partition $\cP^*$ of $\cM$ which is a minimizer in \eqref{eq:I}, we have 
\begin{align}
\textbf{I}(X_{\mathcal{M}}^n)-\textbf{I}(X_{\mathcal{M}}^n|\textbf{L}) 
    &\leq I(X_{\mathcal{M}}^n;\textbf{L})-\sum_{B\in\mathcal{B}}\lambda_B^{(\cP^*)}I(X_B^n;\textbf{L}|X_{B^c}^n) \label{prop:ineqs:1}\\
   &=H(\textbf{L})-\sum_{B\in\mathcal{B}}\lambda_B^{(\cP^*)}H(\textbf{L}|X_{B^c}^n) \notag
\end{align}
where \eqref{prop:ineqs:1} follows from \eqref{lambdap} and \eqref{cmi} and hence, 
\begin{equation}
\frac{1}{n}H(\textbf{L})\geq\textbf{I}(X_{\mathcal{M}})-\frac{1}{n}\textbf{I}(X_{\mathcal{M}}^n|\textbf{L}).
\label{eq:a}
\end{equation} 
Now, if $\textbf{L}$ is any $\CI_W$ of rate $R$, then by Definitions~\ref{def:CI} and \ref{def:CIrate}, for every $\epsilon>0$, we have $\frac{1}{n}H(\textbf{L}) \le R + \epsilon$ and $\frac{1}{n}\textbf{I}(X_{\mathcal{M}}^n|\textbf{L})\leq \epsilon$ for all sufficiently large $n$. Thus, in conjunction with \eqref{eq:a}, we have $R + \epsilon \ge \frac{1}{n}H(\textbf{L}) \ge \textbf{I}(X_{\mathcal{M}})-\epsilon$ for all sufficiently large $n$. In particular, $R+\epsilon \ge \textbf{I}(X_{\mathcal{M}})-\epsilon$ holds for any $\epsilon > 0$, from which we infer that $R \ge \textbf{I}(X_{\mathcal{M}})$. The inequality $\CI_W(X_{\mathcal{M}}) \ge \textbf{I}(X_{\mathcal{M}})$ now follows.
\end{IEEEproof}

Finally, analogous to Definition~\ref{def:RSKr}, we have a definition of achievable rate of interactive communication required to get a $\CI$. 

\begin{definition}
\label{def:RCI}
A real number $R\geq 0$ is said to be an \emph{achievable rate of interactive communication for $\CI$} if for all $\epsilon > 0$ and for all sufficiently large $n$, there exist \emph{(i)}~an interactive communication $\textbf{F}$ satisfying $\frac{1}{n}\log|\cF| \; \leq R+\epsilon$, and \emph{(ii)}~a CR $\textbf{J}$ such that $\textbf{L}=(\textbf{J},\textbf{F})$ is a $\CI$. We denote the infimum among all such achievable rates by $R_{\CI}$.
\end{definition}

\section{Lower Bound on $R_{\SK}$} \label{sec:commcomp}

The goal of this section is to state and prove a lower bound on $R_{\SK}$, which partially extends Tyagi's two-terminal result \cite[Theorem~3]{Tyagi13} to the multiterminal setting.

\begin{theorem}
For any multiterminal source $X_{\cM}$, we have
$$
R_{\SK}\geq R_{\CI}\geq \CI(X_{\mathcal{M}})-\textbf{I}(X_{\mathcal{M}}).
$$
\label{th:commcomp}
\end{theorem}
\vspace*{-10pt} 
By Proposition~\ref{prop:ineqs}, the lower bounds above are non-negative.

The ideas in our proof of Theorem~\ref{th:commcomp} may be viewed as a natural extension of those in the proof of \cite[Theorem 3]{Tyagi13}. We start with three preliminary lemmas. In all that follows, $\lambda^* = \lambda^{(\cP^*)}$ for any $\cP^*$ that achieves the maximum in the right-hand side of \eqref{cmi}; moreover, $\lambda^* = (\lambda^*_B: B \in \mathcal{B})$.

\begin{lemma}
For any function $\textbf{L}$ of $X_{\mathcal{M}}$, we have 
$$
n\textbf{I}(X_{\mathcal{M}}) = \textbf{I}(X_{\mathcal{M}}^n|\textbf{L})+H(\textbf{L})-\sum_{B\in \mathcal{B}} \lambda_B^*H(\textbf{L}|X_{B^c}^n). 
$$
\label{lemma:eq}
\end{lemma}

\begin{IEEEproof}
Consider $\textbf{L}=\textbf{L}(X_{\mathcal{M}}^n)$. From \eqref{skcapacity}, we have
\begin{align*}
n\textbf{I}(X_{\mathcal{M}}) & = H(X_{\mathcal{M}}^n)-\sum_{B\in \mathcal{B}} \lambda_B^* H(X_B^n|X_{B^c}^n) \nonumber \\
                                                & = H(X_{\mathcal{M}}^n,\textbf{L})-\sum_{B\in \mathcal{B}} \lambda_B^* H(X_B^n,\textbf{L}|X_{B^c}^n) \\
                                                & = H(X_{\mathcal{M}}^n|\textbf{L}) + H(\textbf{L})-\sum_{B\in \mathcal{B}} \lambda_B^*H(X_B^n|\textbf{L},X_{B^c}^n) -\sum_{B\in \mathcal{B}}\lambda_B^*H(\textbf{L}|X_{B^c}^n) \\
                                                &= \textbf{I}(X_{\mathcal{M}}^n|\textbf{L})+H(\textbf{L})-\sum_{B\in \mathcal{B}} \lambda_B^*H(\textbf{L}|X_{B^c}^n) 
\end{align*}
the last equality above being due to \eqref{lambdap} and \eqref{cmi}.
\end{IEEEproof}

\begin{lemma}
For any CR $\textbf{J}$ obtained from an interactive communication $\textbf{F}$,
$$
\lim_{n \to \infty} \frac{1}{n}\sum_{B\in \mathcal{B}} \lambda_B^* H(\textbf{J}|X_{B^c}^n,\textbf{F}) = 0.
$$
\label{lem:JF}
\end{lemma}
\begin{IEEEproof}
Fix an $\epsilon > 0$. We have for all sufficiently large $n$, by Fano's inequality,
\begin{align}
\frac{1}{n}\sum_{B\in \mathcal{B}} \lambda_B^* H(\textbf{J}|X_{B^c}^n,\textbf{F})
           & \leq \frac{1}{n}\sum_{B\in \mathcal{B}}\lambda_B^* \left( h(\epsilon)+\epsilon H(X_{B^c}^n,\textbf{F})\right) \nonumber \\
           &  \leq \frac{1}{n}\sum_{B\in \mathcal{B}}\lambda_B^* \left( h(\epsilon)+\epsilon H(X_{\mathcal{M}}^n,\textbf{F})\right) \nonumber \\
           &  = \frac{1}{n}\sum_{B\in \mathcal{B}}\lambda_B^* \left( h(\epsilon)+\epsilon H(X_{\mathcal{M}}^n)\right) \nonumber \\
           &  = \frac{1}{n}\sum_{B\in \mathcal{B}}\lambda_B^* \left( h(\epsilon)+n \epsilon H(X_{\mathcal{M}})\right) \nonumber \\
           &  \leq (2^m-2)\left[h(\epsilon) + \epsilon  H(X_{\mathcal{M}}) \right] \label{th:commcomp:a} 
\end{align}
where $h(.)$ is the binary entropy function, and (\ref{th:commcomp:a}) follows from the fact that, by definition, $\lambda_B^*\leq 1$ and $|\mathcal{B}|=2^m-2$. Note that the expression in \eqref{th:commcomp:a} goes to $0$ with $\epsilon$, since $h(\epsilon)\to 0$ as $\epsilon \to 0$, and $H(X_{\mathcal{M}}) \le \log(\prod_{j=1}^m |\mathcal{X}_j|)$.
\end{IEEEproof}

The last lemma we need, stated without proof, is a special case of \cite[Lemma B.1]{CN08}.

\begin{lemma}[\cite{CN08}, Lemma B.1]
For an interactive communication $\textbf{F}$ we have
$$
H(\textbf{F})\geq \sum_{B\in \mathcal{B}} \lambda_B^*H(\textbf{F}|X_{B^c}^n).
$$
\label{lem:comm}
\end{lemma}

\vspace*{-15pt}
With these lemmas in hand, we can proceed to the proof of Theorem~\ref{th:commcomp}.

\begin{IEEEproof}[Proof of Theorem \ref{th:commcomp}]
The proof is done in two parts. In the first part, we prove that $R_{\CI}\geq \CI(X_{\mathcal{M}})-\textbf{I}(X_{\mathcal{M}})$. In the second part, we show that $R_{\SK}\geq R_{\CI}$.

\textit{Part~I: $R_{\CI} \geq \CI(X_{\mathcal{M}})-\textbf{I}(X_{\mathcal{M}})$}

The idea is to show that $\textbf{I}(X_{\mathcal{M}})+R_{\CI}$ is an achievable $\CI$ rate, so that $\CI(X_{\mathcal{M}}) \le \textbf{I}(X_{\mathcal{M}})+R_{\CI}$.

Fix an $\epsilon > 0$. By the definition of $R_{\CI}$, for all sufficiently large $n$, there exists an interactive communication $\textbf{F}$ satisfying $\frac{1}{n}\log |\cF| \; \leq R_{\CI}+\epsilon/2$ and a CR $\textbf{J}$ such that $\textbf{L}=(\textbf{J},\textbf{F})$ is a $\CI$. We will show that $\frac{1}{n} H(\textbf{J},\textbf{F}) \le \textbf{I}(X_{\mathcal{M}})+R_{\CI}+\epsilon$ for all sufficiently large $n$. This, by Definition~\ref{def:CIrate}, shows that $\textbf{I}(X_{\mathcal{M}})+R_{\CI}$ is an achievable $\CI$ rate.

Setting $\textbf{L} = (\textbf{J},\textbf{F})$ in Lemma~\ref{lemma:eq}, we obtain
\begin{align}
\frac{1}{n}\left[H(\textbf{J},\textbf{F})  -\sum_{B\in \mathcal{B}} \lambda_B^* H(\textbf{F}|X_{B^c}^n)\right] -\textbf{I}(X_{\mathcal{M}}) 
& 
= \ \frac{1}{n}\left[\sum_{B\in \mathcal{B}} \lambda_B^* H(\textbf{J}|X_{B^c}^n,\textbf{F})-\textbf{I}(X_{\mathcal{M}}^n|\textbf{J},\textbf{F})\right] \nonumber \\
& 
\leq \ \epsilon/2, \label{th:commcomp:c}
\end{align}
where (\ref{th:commcomp:c}) follows from Lemma~\ref{lem:JF}. Re-arranging, we get
\begin{align*}
\frac{1}{n} H(\textbf{J},\textbf{F}) & \le \textbf{I}(X_{\mathcal{M}})+\frac{1}{n} \sum_{B\in \mathcal{B}} \lambda_B^* H(\textbf{F}|X_{B^c}^n) +\epsilon/2 \\
& \le \textbf{I}(X_{\mathcal{M}})+\frac{1}{n} H(\textbf{F})+\epsilon/2
\end{align*}
the second inequality coming from Lemma~\ref{lem:comm}. Finally, using the fact that $\frac{1}{n}H(\textbf{F})\leq \frac{1}{n}\log|\cF| \; \leq R_{\CI}+\epsilon/2$, we see that
$$
\frac{1}{n} H(\textbf{J},\textbf{F}) \leq \textbf{I}(X_{\mathcal{M}})+R_{\CI}+\epsilon 
$$
which is what we set out to prove.

\textit{Part II: $R_{\SK}\geq R_{\CI}$}

Fix $\epsilon > 0$. From the definition of $R_{SK}$, there exist an interactive communication $\textbf{F}$ and an SK $\textbf{K}$ obtained from $\textbf{F}$ such that, for all sufficiently large $n$, $\frac{1}{n}\log|\cF| \; \leq R_{SK}+\epsilon$ and $\frac{1}{n}H(\textbf{K})\geq \textbf{I}(X_{\mathcal{M}})-\epsilon$. We wish to show that $(\textbf{K},\textbf{F})$ is a $\CI$, so that by Definition~\ref{def:RCI}, we would have  $R_{\SK}\geq R_{\CI}$.

Setting $\textbf{L}=(\textbf{K},\textbf{F})$ in Lemma~\ref{lemma:eq}, we have for all sufficiently large $n$,
\begin{align}
\frac{1}{n}\textbf{I}(X_{\mathcal{M}}^n|\textbf{K},\textbf{F}) 
            & = \textbf{I}(X_{\mathcal{M}})-\frac{1}{n}H(\textbf{K},\textbf{F})+\frac{1}{n}\sum_{B\in \mathcal{B}} \lambda_B^*H(\textbf{F}|X_{B^c}^n)+\frac{1}{n}\sum_{B\in \mathcal{B}}\lambda_B^*H(\textbf{K}|X_{B^c}^n,\textbf{F}) \nonumber \\
            & \leq  \textbf{I}(X_{\mathcal{M}})-\frac{1}{n}H(\textbf{K}|\textbf{F}) + \epsilon \label{th:commcomp:e} \\
            & \leq   \textbf{I}(X_{\mathcal{M}})-\frac{1}{n}H(\textbf{K})+\epsilon+\epsilon \label{th:commcomp:f} \\ 
            &  \leq 3\epsilon, \label{th:commcomp:g}
\end{align}
where (\ref{th:commcomp:e}) follows from Lemmas~\ref{lem:JF} and \ref{lem:comm}, (\ref{th:commcomp:f}) follows from the fact that $I(\textbf{K};\textbf{F})\leq \epsilon$, while (\ref{th:commcomp:g}) is due to the fact that $\frac{1}{n}H(\textbf{K})\geq \textbf{I}(X_{\mathcal{M}})-\epsilon$. Thus, by Definition~\ref{def:CI}, $(\textbf{K},\textbf{F})$ is a $\CI$.
\end{IEEEproof}

An issue with our Theorem~\ref{th:commcomp} is that the bounds are difficult to evaluate explicitly, as we do not have a computable characterization of $\CI(X_{\mathcal{M}})$. In addition to that, we do not know if the lower bounds of Theorem~\ref{th:commcomp} are in general tight, in the sense of there being matching upper bounds. For the special case of the two-terminal model, Theorem~3 of \cite{Tyagi13} shows that the bound on $R_{\SK}$ is tight (albeit under a weaker notion of SK, as explained in Footnote~\ref{fn:weakSK}). In the general multiterminal model, with $m\geq 3$, the best known upper bound on $R_{\SK}$ is the minimum rate of communication for omniscience, $R_{\CO}$. In the following section, we identify a large class of sources where our lower bound equals $R_{\CO}$, i.e., the sources are $R_{\SK}$-maximal.

\section{$R_{\SK}$-maximality in uniform hypergraph PIN models} \label{sec:RSKmax}

This section focuses on a special class of sources called the PIN model, introduced in \cite{NYBNR10} and \cite{NN10}. A broad class of PIN models defined on uniform hypergraphs (which is a generalization of the PIN models of \cite{NYBNR10} and \cite{NN10}) is identified to be $R_{\SK}$-maximal in this section. Briefly, a \emph{hypergraph PIN model} is defined on an underlying hypergraph $\cH = (\cV,\cE)$ with $\cV = \cM$, the set of $m$ terminals of the model, and $\cE$ being a \emph{multiset} of hyperedges, i.e., subsets of $\cV$.\footnote{Note that we allow $\cE$ to contain multiple copies of a hyperedge. In the graph theory literature, such a hypergraph is sometimes referred to as a ``multi-hypergraph''.} For a hyperedge $e$ having $\ell$ copies in the multiset $\cE$, we represent the different copies as $e_1,e_2,\ldots,e_{\ell}$. To keep the notation simple, if a hyperedge $e$ has only one copy in $\cE$, we simply represent it as $e$ instead of $e_1$. Unless otherwise stated, we will assume that each hyperedge in $\cE$ has only one copy. For $n \in \N$, we define $\cE^{(n)}$ to be the multiset of hyperedges formed by taking $n$ copies of each element of the multiset $\cE$. Associated with each hyperedge $e \in \cE^{(n)}$ is a Bernoulli$(1/2)$ random variable $\xi_e$; the $\xi_e$s are all mutually independent. With this, the random variables $X_i^n$, $i\in\mathcal{M}$, are defined as $X_i^n=(\xi_e$ : $e\in\mathcal{E}^{(n)}$ and $i\in e$). When every $e\in\cE$ satisfies ${|e|}=t$, we call $\cH$ a \emph{$t$-uniform hypergraph}. We will show that any Type $\cS$ uniform hypergraph PIN model is $R_{\SK}$-maximal.

\begin{theorem}
For a Type $\cS$ PIN model defined on an underlying $t$-uniform hypergraph $\cH = (\cV,\cE)$, we have $\CI(X_{\cM})=\CI_W(X_{\cM})=H(X_{\cM})$, and hence, $R_{\SK}=R_{\CO} = \frac{m-t}{m-1} |\cE|$.
\label{th:uh}
\end{theorem}

Type $\cS$ PIN models defined on $t$-uniform hypergraphs do indeed exist, as we will see in Section~\ref{sec:singleton}. Also, it is possible to efficiently determine if a given source $\cX_{\cM}$ (not necessarily a PIN model) is Type $\cS$; a strongly polynomial-time algorithm for this has been given by Chan et al.\ \cite{Chan14}. In Section~\ref{sec:singleton}, we present another useful, but inefficient, test for deciding the Type $\cS$ property.

The proof of Theorem~\ref{th:uh} will require a technical lemma which we state below.

\begin{lemma}
\label{lem:mi}
For any \emph{$t$-uniform hypergraph} PIN model and any function $\textbf{L}$ of $X_{\mathcal{M}}^n$ we have
\begin{equation}
\sum_{i=1}^m I(X_i^n;\textbf{L}) \leq tH(\textbf{L}). \label{eq:mi}
\end{equation}
\end{lemma}

The lengthy proof of this lemma is deferred to Appendix~\ref{app:proof}.

\begin{IEEEproof}[Proof of Theorem \ref{th:uh}]
Observe that $\lambda^{(\cS)}_B=\frac{1}{m-1}$, whenever ${|B|}=m-1$ and $\lambda^{(\cS)}_B=0$, otherwise. Hence, for any Type $\cS$ source $X_{\cM}^n$, we have
\begin{equation}
\textbf{I}(X_{\mathcal{M}}^n | \textbf{L}) \geq H(X_{\mathcal{M}}^n | \textbf{L})- \frac{1}{m-1} \sum_{i=1}^m H(X_{\cM \setminus \{i\}}^n|X_{i}^n,\textbf{L})  \label{cmi_uh}
\end{equation}
using \eqref{lambdap} and \eqref{cmi}.
Now assume that $X_{\cM}$ arises from a PIN model defined on a $t$-uniform hypergraph $\cH = (\cV,\cE)$, and consider any function $\textbf{L}$ of $X_{\mathcal{M}}^n$. This allows us to further simplify \eqref{cmi_uh}:
\begin{align}
\textbf{I}(X_{\mathcal{M}}^n | \textbf{L}) 
   & \geq H(X_{\mathcal{M}}^n) - H(\textbf{L}) - \frac{1}{m-1} \sum_{i=1}^m \left[H(X_{\cM}^n) - H(X_i^n) - H(\textbf{L}|X_i^n)\right] \notag \\
   & = \frac{n(t-1)|\cE|}{m-1} - H(\textbf{L}) + \frac{1}{m-1} \sum_{i=1}^m H(\textbf{L}|X_i^n) \label{cmi_uh2}\\
   & = \frac{n(t-1)|\cE|}{m-1} - \frac{1}{m-1} \left[ \sum_{i=1}^m I(X_i^n;\BL)-H(\BL)\right] \notag \\
   & = \frac{n(t-1)}{m-1}\left(|\cE|-\frac{1}{n}H(\BL)\right)-\frac{1}{m-1} \left[ \sum_{i=1}^m I(X_i^n;\BL)-tH(\BL)\right] \notag \\
   & \geq \frac{n(t-1)}{m-1}\left(|\cE|-\frac{1}{n}H(\BL)\right), \label{cmi:uh3}
\end{align}
the equality \eqref{cmi_uh2} using the facts that $H(X_{\cM}^n) = n|\cE|$ and $\sum_{i=1}^m H(X_i^n) = nt|\cE|$, and \eqref{cmi:uh3} following from Lemma~\ref{lem:mi}.

We will now compute $\CI(X_{\cM})$ using Proposition \ref{prop:ineqs}. The upper bound gives us $\CI(X_{\cM}) \le |\cE|$, as $H(X_{\cM})=|\cE|$. For the lower bound, let $\textbf{L}$ be any $\CI_W$ so that for any $\epsilon > 0$, we have $\frac{1}{n} \textbf{I}(X_{\mathcal{M}}^n | \textbf{L}) < \frac{(t-1)\epsilon}{(m-1)}$ for all sufficiently large $n$. The bound in \eqref{cmi:uh3} thus yields $\frac{1}{n}H(\BL)> |\cE|-\epsilon$ for all sufficiently large $n$. Hence, it follows that $\CI_W(X_{\cM}) \ge |\cE|$. From the upper and lower bounds in Proposition \ref{prop:ineqs}, we then obtain $CI_W(X_{\cM})=CI(X_{\cM})=H(X_{\cM})$.

Now from Theorem \ref{th:commcomp} we have $R_{\SK}\geq CI(X_{\cM})-\textbf{I}(X_{\cM})$. Hence, we have
\begin{gather}
R_{\SK} \ge |\cE|- \textbf{I}(X_{\cM})= H(X_{\cM})-\textbf{I}(X_{\cM})= R_{\CO}, \label{cmi:uh4}
\end{gather}
where the last equality is from \eqref{omni}. But we also have $R_{\SK}\leq R_{\CO}$, as pointed out in Section \ref{sec:prelim}, which proves that $R_{\SK} = R_{\CO}$. 

To obtain the exact expression for $R_{\CO}$, we note that by \eqref{omni} and \eqref{eq:I}, $R_{\CO} = H(X_{\cM}) - \Delta(\cS) = \frac{m}{m-1} H(X_{\cM}) - \frac{1}{m-1} \sum_{i=1}^m H(X_i)$. This simplifies to the expression stated in the theorem using the facts (already mentioned above) that $H(X_{\cM}) = |\cE|$ and $\sum_{i=1}^m H(X_i) = t |\cE|$.
\end{IEEEproof} 

It turns out that for PIN models on graphs (i.e., $t=2$), the Type $\cS$ condition is also necessary for $R_{\SK}$-maximality. It is possible that this holds for PIN models on $t$-uniform hypergaphs (with $t \ge 3$) as well, but we do not have a proof for this yet. 

\begin{theorem}
A PIN model defined on a graph is $R_{\SK}$-maximal iff it is Type $\cS$.
\label{th:graph}
\end{theorem}

We will prove the necessity of the Type $\cS$ condition by showing that any graph PIN model that is not Type $\cS$ has an SK-capacity-achieving protocol of communication rate strictly less than $R_{\CO}$. To do this, we need a few preliminaries. Consider a graph $\cG=(\cV,\cE)$ and define $\cG^{(n)}=(\cV,\cE^{(n)})$ for any positive integer $n$. The \emph{spanning tree packing number} of $\cG^{(n)}$, denoted by $\sigma(\cG^{(n)})$, is the maximum number of edge-disjoint spanning trees of $\cG^{(n)}$. It is a fact that $\displaystyle\lim_{n\to\infty}\frac{1}{n}\sigma(G^{(n)})$ exists (see \cite[Proposition~4]{NN10}); we denote this limit by $\overline{\sigma}(\cG)$ and call it the \emph{spanning tree packing rate} of the graph $\cG$. It was shown in \cite[Theorem 5]{NN10} that for a PIN model on $\cG$, we have $\cC(\cM)=\overline{\sigma}(\cG)$. Therefore, by \eqref{omni} we have $R_{\CO}=H(X_{\cM})-\overline{\sigma}(\cG)={|\cE|}-\overline{\sigma}(\cG)$. We also have the following lemma, the proof of which is given in Appendix~\ref{app:tree}.

\begin{lemma}
For a PIN model defined on a graph $\cG$, we have $R_{\SK}\leq (m-2)\overline{\sigma}(\cG)$.
\label{lem:tree}
\end{lemma}

\begin{IEEEproof}[Proof of Theorem~\ref{th:graph}]
The ``if'' part follows from Theorem~\ref{th:uh}. For the ``only if'' part, consider a PIN model on $\cG$ that is not of Type $\cS$. Using the fact that SK capacity equals the spanning tree packing rate, we then have via \eqref{eq:I}
\begin{gather*}
\overline{\sigma}(\cG) = \cC(\cM) < \Delta(\cS)=\frac{|\cE|}{m-1}.
\end{gather*}
Therefore, ${|\cE|} > (m-1)\overline{\sigma}(\cG)$, or equivalently, $|\cE| - \overline{\sigma}(\cG) > (m-2) \overline{\sigma}(\cG)$. Since $R_{\CO} = |\cE| - \overline{\sigma}(\cG)$, we obtain $R_{\CO} > R_{\SK}$ via Lemma~\ref{lem:tree}.
\end{IEEEproof} 

It is natural to ask at this point whether all Type $\cS$ sources (not necessarily PIN models) are $R_{\SK}$-maximal. The answer turns out to be ``No'', as shown by the following example. 

\begin{example}
Let $W$ be a Ber($p$) rv, for some $p \in [0,1]$: $\Pr[W = 1] = 1 - \Pr[W=0] = p$. Let $X_1,\ldots,X_m$ be random variables that are conditionally independent given $W$, with 
$$\Pr[X_i = 01 |  W = 0] = 1 -  \Pr[X_i = 00 |  W = 0] = 0.5$$ 
and 
$$\Pr[X_i = 11 |  W = 1] = 1 - \Pr[X_i = 10 |  W = 1] = 0.5$$
for $i = 1,2,\ldots,m$. Denote by $h(p)$ the binary entropy of $p$.

It is easy to check that $H(X_A)={|A|}+h(p)$  for all $A\subseteq\cM$, and $H(X_i | X_j) = 1$ for all distinct $i,j \in\cM$. Therefore, all partitions $\cP$ of $\cM$ satisfy $\Delta(\cP)=h(p)$, and hence, $\textbf{I}(X_{\cM})=h(p)$. In particular, $X_{\cM}$ defines a Type $\cS$ source. Furthermore, using \eqref{omni}, we have $R_{\CO}=m$. 

We now show that $R_{\SK}<R_{\CO}$. Consider a Slepian-Wolf code (see \cite[Section 10.3.2]{ElK11}) of rate $H(X_1| X_2)=1$ for terminal 1. All terminals can recover $X_1^n$ since $H(X_1 | X_i)=1$ for all $i\in \{2,3,\cdots,m\}$. Then, using the balanced coloring lemma \cite[Lemma B3]{CN04} on $X_1^n$, an SK of rate $H(X_1)-H(X_1| X_2)=h(p)$ can be obtained. Hence, $R_{\SK}\leq 1<m=R_{\CO}$.
\label{ex:omni}
\end{example}

In fact, there exist non $R_{\SK}$-maximal sources with $\cS$ being a \emph{unique} minimizer for \eqref{eq:I}. We provide one such example in Appendix~\ref{app:exunique}.

\section{Omnivocality: When is it necessary?}\label{sec:omnivocal}

It is a well-established fact (see \cite{Maurer93}, \cite{AC93}) that to generate a maximal-rate SK within a two-terminal source model, it is enough for only one terminal to communicate.\footnote{To be precise, the results in \cite{Maurer93} and \cite{AC93} are based on a weaker notion of secrecy, where in Definition~\ref{def:SK}, the condition $I(\BK;\BF)\leq\epsilon$ is replaced by $\frac{1}{n}I(\BK;\BF)\leq\epsilon$. However, it can be shown that one terminal communicating suffices to achieve SK capacity for $m=2$, in the stronger sense as in Definition~\ref{def:SK}. Terminal 1 uses a Slepian-Wolf code of rate $H(X_1|X_2)$ to communicate $X_1^n$ to terminal 2. Both terminals now use a balanced coloring (see \cite[Lemma B.3]{CN04}) on $X_1^n$ to get a strong SK of rate $I(X_1;X_2)=\textbf{I}(X_1,X_2)$.\label{foot:weak}} So it is natural to ask whether this fact extends to the general multiterminal setting. In other words, for $m\geq 3$, is there always an SK generation protocol involving $m-1$ or fewer terminals communicating that achieves SK capacity? If not, can we identify a class of sources where omnivocality, i.e., all terminals communicating, is required to achieve SK capacity? This section addresses these questions. The main result of this section says that if a source is strict Type $\cS$, then omnivocality is required for achieving SK capacity.

\begin{theorem}
For a strict Type $\cS$ source on $m \ge 3$ terminals, omnivocal communication is necessary for achieving SK capacity. 
\label{th:mge3}
\end{theorem}
There indeed exist sources which are strict Type $\cS$. We give a few examples of such sources in Section~\ref{sec:singleton}. 

Theorem~\ref{th:mge3} gives a sufficient condition for identifying sources where omnivocality is necessary to generate a maximal-rate SK. The next result shows that the condition is also necessary when $m=3$, i.e., for any source on 3 terminals which is not strict Type $\cS$, there always exists a non-omnivocal key generation protocol that leads to SK capacity. 

\begin{theorem}
In the three-terminal source model, omnivocal communication is necessary for achieving SK capacity iff the singleton partition $\cS$ is the unique minimizer for $\textbf{I}(X_{\cM})$ in \eqref{eq:I}.
\label{th:meq3}
\end{theorem}

A conjecture was made in \cite{MKS14} that the necessity of omnivocality implies a strict Type $\cS$ source for any $m\geq 3$. It turns out that the conjecture is incorrect. Chan et al.\ have found an explicit example \cite[Example~C.1]{Chan14} of a non-strict Type $\cS$ PIN model on $m>3$ terminals that requires omnivocality to achieve SK capacity. For ease of reference, we reproduce this example in Appendix~\ref{app:chan}. Theorem~5 in \cite{ZLL15} indicates the existence of other such examples.

We now turn to the proofs of Theorems~\ref{th:mge3} and \ref{th:meq3}. We prove the former theorem first. The main technical result used in the proof is the SK capacity with silent terminals by Gohari and Anantharam in \cite[Theorem 6]{GA10}. More precisely, suppose we restrict ourselves to SK generation protocols where, only an arbitrary subset of terminals $T\subset \cM$ is allowed to communicate. We denote the maximum rate of SK that can be generated by such protocols by $\Itxm$. Then we have\footnote{Theorem 6 of \cite{GA10} was based on the weaker notion of secrecy pointed out in Footnote~\ref{foot:weak}. However, it can be easily verified that the result is still valid for the stronger notion of secrecy as in Definition~\ref{def:SK}.}

\begin{theorem}[Theorem~6, \cite{GA10}]
For any $T\subset\cM$, $\Itxm = H(X_{T})-R_{T}^{(\min)}$, where $\RTmin=\displaystyle\min_{\sfR\in\cR_{T}} \sum_{i \in T} R_i$, with 
\begin{equation}
\cR_{T}=\biggl\{\sfR = (R_i, i \in T) : \sum_{i\in B\cap T} R_i\geq H(X_{B\cap T}|X_{B^c}),\ \forall \, B\subsetneq\cM, \ B\cap T\neq\emptyset\biggr\}.
\end{equation}
\label{th:silent}
\end{theorem}

Note that if $\Ixm > \Itxm$ for all $T \subset \cM$ of size ${|T|}=m-1$, then omnivocality is necessary for achieving SK capacity. Thus, our approach for showing that omnivocal communication is needed in certain cases is to use Theorem~\ref{th:silent} to prove that $\Ixm > \Itxm$ for all $(m-1)$-subsets $T \subset \cM$. For this, we will need a lower bound on $\RTmin$ when ${|T|} = m-1$. To prove this bound, we make use of a simpler characterization (than that given in Theorem~\ref{th:silent}) of the rate region $\cR_T$ when ${|T|}=m-1$.

\begin{lemma} Let $T = \cM \setminus \{u\}$ for some $u \in \cM$. 
The rate region $\cR_{T}$ is the set of all points $(R_i, \, i \in T)$ such that
\begin{align}
\displaystyle\sum_{i \in B} R_i & \ge H(X_{B} | X_{T \setminus B}) \quad \forall \, B \subsetneq T, B \ne \emptyset, \text{ and} \label{eq:lem}\\
\displaystyle \sum_{i\in T} R_i & \ge H(X_T| X_u). \notag
\end{align}
\label{lem:Rm-1}
\end{lemma}
\begin{IEEEproof}
Observe that $\cR_T$ is defined by constraints on sums of the form $\sum_{i \in B'} R_i$ for non-empty subsets $B' \subseteq T$. When $B' = T$, the constraint is simply $\sum_{i \in T} R_i \ge H(X_T| X_u)$.

Now, consider any non-empty $B' \subsetneq T$. From Theorem~\ref{th:silent}, we see that constraints on $\sum_{i \in B'} R_i$ arise as constraints on $\sum_{i \in B \cap T} R_i$ in two ways: when $B = B'$ and when $B = B' \cup \{u\}$.  Thus, we have two constraints on $\sum_{i \in B'} R_i$: 
$$
\sum_{i \in B'} R_i \ge H(X_{B'} | X_{\cM \setminus B'}),
$$
obtained when $B = B'$, and 
$$
\sum_{i \in B'} R_i \ge H(X_{B'} | X_{T \setminus B'}),
$$
obtained when $B = B' \cup \{u\}$. The latter constraint is clearly stronger, so we can safely discard the former.
\end{IEEEproof}

We can now prove the desired lower bound on $\RTmin$.

\begin{lemma}
Let $m \ge 3$ be given. For $T \subset \cM$ with ${|T|} = m-1$, we have 
$$
\RTmin \ge \frac{1}{m-2} \sum_{j \in T} H(X_{T \setminus \{j\}} | X_j).
$$
\label{lem:RTmin}
\end{lemma}
\begin{IEEEproof}
Consider any $T \subset \cM$  with ${|T|} = m-1$. For each $j \in T$, let $B_j = T \setminus \{j\}$. Now, let $(R_i, i \in T)$ be any point in $\cR_T$. Applying \eqref{eq:lem} with $B = B_j$, we get
$$
\sum_{i \in B_j} R_i \ge H(X_{T\setminus \{j\}} | X_j),
$$
for each $j \in T$. Summing over all $j \in T$, we obtain
\beq
\sum_{j \in T} \sum_{i \in B_j} R_i \ge \sum_{j \in T} H(X_{T\setminus \{j\}} | X_j).
\label{RTmin_eqa}
\eeq
Exchanging the order of summation in the double sum on the left-hand side (LHS) above, we have
\begin{gather}
\sum_{j \in T} \sum_{i \in B_j} R_i = \sum_{i \in T} \sum_{j \in B_i} R_i = \sum_{i \in T} (m-2)R_i \ = \ (m-2) \sum_{i \in T} R_i. \notag
\end{gather}
Putting this back into \eqref{RTmin_eqa}, we get 
$$
\sum_{i \in T} R_i \ge \frac{1}{m-2} \sum_{j \in T} H(X_{T\setminus \{j\}} | X_j).
$$
Since this holds for any point $(R_i, i \in T) \in \cR_T$, the lemma follows.
\end{IEEEproof}  

For the proof of Theorem~\ref{th:mge3}, we need some convenient notation. For $T \subset \cM$, ${|T|} = m-1$, define $\Delta_T(\cS) \triangleq \frac{1}{m-2}[\sum_{i \in T} H(X_i) - H(X_T)]$.

\begin{lemma}
For $m \ge 3$ terminals, if the singleton partition $\cS$ is the unique minimizer for $\Ixm$, then $\Delta_T(\cS)  <  \Delta(\cS)$ for all $T \subset \cM$ with ${|T|} = m-1$.
\label{lem:Delta}
\end{lemma}
\begin{IEEEproof}
For any $u \in \cM$, consider $T = \cM \setminus \{u\}$. Using $\Delta(\cS) = \frac{1}{m-1} [\sum_{i=1}^m H(X_i) - H(X_{\cM})]$ and the definition of $\Delta_T(\cS)$ above, it is easy to verify the identity
$$
{\textstyle
\frac{m-1}{m-2} \Delta(\cS) = \Delta_T(\cS) + \frac{1}{m-2} I(X_u;X_T).
}
$$
Re-arranging the above, we obtain
\begin{align}
\Delta_T(\cS) - \Delta(\cS) & = {\textstyle \frac{1}{m-2}}[\Delta(\cS) - I(X_u;X_T)] \notag \\
& = {\textstyle \frac{1}{m-2}} [\Delta(\cS) - \Delta(\cP)], \label{Ddiff}
\end{align}
where $\cP$ is the 2-cell partition $\bigl\{\{u\},T \bigr\}$ of $\cM$. By assumption, the expression in \eqref{Ddiff} is strictly negative.
\end{IEEEproof}

\medskip

With this, we are ready to prove Theorem~\ref{th:mge3}.

\begin{IEEEproof}[Proof of Theorem~\ref{th:mge3}]
We will show that $\Ixm > \Itxm$ for any $T \subset \cM$ with ${|T|}=m-1$. First, note that since $\cS$ is, by assumption, a minimizer for \eqref{eq:I}, we have $\Ixm = \Delta(\cS)$. Next, by Theorem~\ref{th:silent} and Lemma~\ref{lem:RTmin}, we have
\begin{align*}
\Itxm  & \leq H(X_T) - {\textstyle \frac{1}{m-2}} \sum_{i\in T}H(X_{T\setminus\{i\}}|X_i) \notag \\
& = {\textstyle \frac{1}{m-2}} \biggl[(m-2) H(X_T) - \sum_{i\in T} [H(X_T) - H(X_i)] \biggr] \notag \\
& = \Delta_T(\cS).
\end{align*}
Therefore, $\Itxm \leq \Delta_T(\cS) < \Delta(\cS) = \Ixm$, the second inequality coming from Lemma~\ref{lem:Delta}. 
\end{IEEEproof}

\medskip

We conclude this section with the proof of Theorem~\ref{th:meq3}. Note that when $m=3$, \eqref{eq:I} reduces to 
\begin{equation}
\Ixm = \min\bigl\{  I(X_{\{1,2\}};X_3), I(X_{\{1,3\}};X_2), I(X_{\{2,3\}};X_1), \; \Delta(\cS)\bigr\}, \label{eq:I3}
\end{equation}
and so, the unique minimizer condition is equivalent to 
$$\Delta(\cS) < \min\{I(X_{\{1,2\}};X_3), I(X_{\{1,3\}};X_2), I(X_{\{2,3\}};X_1)\}.$$
Note also that $\Delta(\cS) = \frac12[H(X_1) + H(X_2) + H(X_3) - H(X_{\{1,2,3\}})]$.

\medskip

\begin{IEEEproof}[Proof of Theorem~\ref{th:meq3}]
The ``if'' part is by Theorem~\ref{th:mge3}. For the ``only if'' part, suppose that $\Delta(\cS) \ge\min\{I(X_{\{1,2\}};X_3),$ $ I(X_{\{1,3\}};X_2), I(X_{\{2,3\}};X_1)\}$. Then, $\Delta(\cS)$ is either (a)~greater than or equal to at least two of the three terms in the minimum, or (b)~greater than or equal to exactly one term. Up to symmetry, it suffices to distinguish between two cases:

 Case I: $\Delta(\cS) \ge \max\{I(X_{\{1,2\}};X_3), I(X_{\{1,3\}};X_2)\}$

 Case II: $\min\{I(X_{\{1,3\}};X_2), I(X_{\{2,3\}};X_1)\} > \Delta(\cS) \ge I(X_{\{1,2\}};X_3)$ 

\noindent In each case, we demonstrate a capacity-achieving communication in which at least one terminal remains silent.

\medskip

We deal with Case~I first. Observe that $\Delta(\cS)$ can be written as $\frac12[I(X_1;X_2) + I(X_{\{1,2\}};X_3)]$. Thus, the assumption $\Delta(\cS) \ge I(X_{\{1,2\}};X_3)$, upon some re-organization, yields $I(X_1;X_2) \ge I(X_{\{1,2\}};X_3)$, i.e.,
\begin{equation}
I(X_1;X_2) \ge I(X_1;X_3) + I(X_2;X_3 | X_1).
\label{case1:eqa}
\end{equation}
 Similarly, using the identity $\Delta(\cS)=\frac12[I(X_1;X_3) + I(X_{\{1,3\}};X_2)]$ in the assumption $\Delta(\cS) \ge I(X_{\{1,3\}};X_2)$, we obtain $I(X_1;X_3) \ge I(X_{\{1,3\}};X_2)$, i.e.,
\begin{equation}
I(X_1;X_3) \ge I(X_1;X_2) + I(X_2;X_3 | X_1).
\label{case1:eqb}
\end{equation}
The equalities in \eqref{case1:eqa} and \eqref{case1:eqb} can simultaneously hold iff 
\begin{equation}
\begin{gathered}
I(X_1;X_2) = I(X_1;X_3) \ \ \text{ and } \\
I(X_2;X_3|X_1) = 0.
\end{gathered}
\label{case1:eqc}
\end{equation}
From \eqref{case1:eqc}, it is not hard to deduce that the quantities $I(X_{\{1,2\}};X_3)$, $I(X_{\{1,3\}};X_2)$ and $\Delta(\cS)$ are all equal to $I(X_1;X_2)$, and $I(X_{\{2,3\}};X_1)=I(X_1;X_2)+I(X_1;X_3|X_2)\geq I(X_1;X_2)$. In particular, $\textbf{I}(X_{\{1,2,3\}}) = I(X_1;X_2)$.

From the first equality in \eqref{case1:eqc}, we also have $H(X_1|X_2) = H(X_1|X_3)$. Now, it can be shown by a standard random binning argument that there exists a communication from terminal $1$ of rate $H(X_1|X_2) = H(X_1|X_3)$ such that $X_1^n$ is a CR. It then follows from the ``balanced coloring lemma'' \cite[Lemma~B.3]{CN04} that an SK rate of $H(X_1) - H(X_1|X_2) = I(X_1;X_2)$ is achievable. Thus, the SK capacity, $\textbf{I}(X_{\{1,2,3\}}) = I(X_1;X_2)$, is achievable by a communication in which terminals $2$ and $3$ are both silent.

\medskip

Now, consider Case~II, in which we obviously have $\textbf{I}(X_{\{1,2,3\}}) = I(X_{\{1,2\}};X_3)$. The idea here is to show that a valid communication of rate $H(X_{\{1,2\}}|X_3)$ exists in which terminal $3$ is silent and $(X_1^n,X_2^n)$ is a CR. Given this, an application of \cite[Lemma~B.3]{CN04} shows that an SK rate of $H(X_{\{1,2\}}) - H(X_{\{1,2\}}|X_3) = I(X_{\{1,2\}};X_3)$ is achievable. Thus, there is a $\textbf{I}(X_{\{1,2,3\}})$-achieving communication in which terminal $3$ is silent. 

To show that the desired communication exists, we argue as follows. For $i=1,2$, let $R_i$ be the rate at which terminal $i$ communicates. A standard random binning argument shows that an achievable $(R_1,R_2)$ region, with terminal $3$ silent, for a communication intended to allow recoverability of $(X_1^n,X_2^n)$ as CR at all terminals is given by 
\begin{equation}
\begin{gathered}
R_1 \ge H(X_1 | X_2), \ \ R_2 \ge H(X_2|X_1), \\
R_1+R_2 \ge H(X_{\{1,2\}}|X_3).
\end{gathered}
\label{case2:eqa}
\end{equation}

Now, using the assumption in Case~II that $\Delta(\cS) \ge I(X_{\{1,2\}};X_3)$, we will prove that the inequality 
\beq
H(X_1 | X_2) + H(X_2|X_1) \le H(X_{\{1,2\}}|X_3)
\label{case2:eqb}
\eeq
holds. It would then follow from \eqref{case2:eqa} that there exist achievable rate pairs $(R_1,R_2)$ with $R_1 + R_2 = H(X_{\{1,2\}}|X_3)$, thus completing the proof for Case~II. 

So, let us prove \eqref{case2:eqb}. We have $\Delta(\cS) =  \frac12[H(X_1) + H(X_2) + H(X_3) - H(X_{\{1,2,3\}})]$ and $I(X_{\{1,2\}};X_3) = H(X_{\{1,2\}})+H(X_3)-H(X_{\{1,2,3\}})$. Using these expressions in the inequality $\Delta(\cS) \ge I(X_{\{1,2\}};X_3)$, and re-arranging terms, we obtain
$$
\frac12[H(X_1)+H(X_2) - 2H(X_{\{1,2\}})] \ge \frac12[H(X_3) - H(X_{\{1,2,3\}})],
$$
which is equivalent to \eqref{case2:eqb}. This completes the proof of the theorem.
\end{IEEEproof}

\section{Finding the minimizing partition}\label{sec:singleton}

The condition that the singleton partition be a unique minimizer for $\mathbf{I}(X_{\cM})$ plays a key role in the results of Section~\ref{sec:RSKmax} and \ref{sec:omnivocal}. Thus, it would be very useful to have a way of checking whether this condition holds for a given source $X_{\cM}$, $m \ge 3$. The brute force method of comparing $\Delta(\cS)$ with $\Delta(\cP)$ for all partitions $\cP$ with at least two parts requires an enormous amount of computation. Indeed, the number of partitions of an $m$-element set is the $m$th Bell number, $B_m$, an asymptotic estimate for which is $(\log w)^{1/2} w^{m-w}e^w$, where $w = \frac{m}{\log m} \, [1+o(1)]$ is the solution to the equation $m = w \log(w+1)$ \cite[Example~5.4]{Odl95}. The proposition below brings down the number of comparisons required for verifying the unique minimizer condition to a ``mere'' $2^m-m-2$. 

For any non-empty subset $B=\{b_1,b_2,\ldots, b_{|B|}\}$ of $\cM$ with $|B| \ <m$, define $\cP_B\triangleq \bigl\{B^c,\{b_1\},\{b_2\},\ldots,\{b_{|B|}\}\bigr\}$ to be the partition of $\cM$ containing $|B|+1$ cells, of which $|B|$ cells are singletons comprising the elements of $B$. Note that if $|B|\ = m-1$, then $\cP_B = \cS$.
\begin{proposition}
For $m \ge 3$, let $\Omega = \{B \subset \cM: 1 \le |B| \, \le m-2\}$. The singleton partition $\cS$ is \\
\emph{(a)} a minimizer for $\mathbf{I}(X_{\cM})$ iff $\Delta(\cS) \le \Delta(\cP_B)$ $\forall\,B \in \Omega$; \\
\emph{(b)} the unique minimizer for $\mathbf{I}(X_{\cM})$ iff $\Delta(\cS) < \Delta(\cP_B)$ $\forall\,B \in \Omega$.
\label{prop:min}
\end{proposition}
There is in fact a strongly polynomial-time algorithm (see \cite{Chan14}) for determining the minimizing partition of \eqref{eq:I}. However, Proposition~\ref{prop:min} is better suited to the purposes of our work. 
\begin{IEEEproof}[Proof of Proposition~\ref{prop:min}]
We prove~(b); for~(a), we simply have to replace the `$>$' in \eqref{ifeqa} below with a `$\ge$'. 
 
The ``only if'' part is obvious. For the ``if'' part, suppose that $\Delta(\cS) < \Delta(\cP_B)$ for all $B \subset \cM$ with $1 \le |B| \ \le m-2$. Consider any partition $\cP$ of $\cM$, $\cP \ne \cS$, with $|\cP|\ \ge 2$. We wish to show that $\Delta(\cP) > \Delta(\cS)$. 

The following identity can be obtained from the definition of $\Delta(\cP)$ by some re-grouping of terms:
$$
\sum_{A \in \cP} |A^c| \, \Delta(\cP_{A^c}) = (|\cP|-1)[\Delta(\cP) + (m-1)\Delta(\cS)].
$$
Thus, we have
\begin{align}
\Delta(\cP) &= \frac{1}{|\cP|-1} \sum_{A \in \cP} |A^c| \, \Delta(\cP_{A^c}) - (m-1)\Delta(\cS) \notag \\
& > \frac{1}{|\cP|-1} \sum_{A \in \cP} |A^c| \, \Delta(\cS) - (m-1)\Delta(\cS) \label{ifeqa} \\
& = m\Delta(\cS) - (m-1) \Delta(\cS) \ = \ \Delta(\cS). \label{ifeqb} 
\end{align}
The inequality in \eqref{ifeqa} is due to the fact that at least one $A \in \cP$ is not a singleton cell, so that $\cP_{A^c} \ne \cS$, and hence, $\Delta(\cP_{A^c}) > \Delta(\cS)$ by assumption. To verify the first equality in \eqref{ifeqb}, observe that $\sum_{A\in\cP}|A^c| = \sum_{A\in\cP}\sum_{i \notin A} 1 = \sum_{i=1}^m \sum_{A \in \cP: i \notin A} 1 = m(|\cP|-1)$.
\end{IEEEproof}

\medskip

Next, we apply Proposition~\ref{prop:min} to some interesting special cases. Random variables $X_1,X_2,\ldots,X_m$, $m \ge 2$, are called \emph{isentropic} if $H(X_A) = H(X_B)$ for any pair of non-empty subsets $A, B \subseteq \cM$ having the same cardinality. Equivalently, $X_1,\ldots,X_m$ are isentropic if, for all non-empty $A \subseteq \cM$, the entropy $H(X_A)$ depends only on $|A|$. As a result, for disjoint subsets $A,B\subseteq\cM$, conditional entropies of the form $H(X_A|X_B)$ depend only on $|A|$ and $|B|$.

\begin{corollary}
\label{isentropic}
Isentropic random variables form a Type $\cS$ source.
\end{corollary}

The proof involves checking that $\Delta(\cS) \le \Delta(\cP_B)$ holds for all $B \in \Omega$, so that the result follows from Proposition~\ref{prop:min}(a). We defer the details to Appendix~\ref{app:typeS}.

There are many examples of isentropic random variables. For example, exchangeable random variables (cf.\ \cite{TKSV12}) are isentropic. (Random variables $X_1,X_2,\ldots,X_m$ are \emph{exchangeable} if for every permutation $\Pi:\cM\to\cM$, the distribution of $X_{\Pi(1)},X_{\Pi(2)},\ldots,X_{\Pi(m)}$ remains unchanged.) A more relevant example for us is the PIN model defined on the \emph{complete $t$-uniform hypergraph on $m$ vertices}, $K_{m,t}$. More precisely, the complete $t$-uniform hypergraph $K_{m,t}=(\cV,\cE)$ has $\cV=\cM$, and exactly one copy of every $t$-subset (i.e., subset of cardinality $t$) of $\cM$ belongs to $\cE$. It is straightforward to check that the random variables $X_1,X_2,\ldots,X_m$ in the PIN model on $K_{m,t}$ are isentropic, and hence the source is Type $\cS$. In fact, we will show below that this PIN model is strict Type $\cS$, and therefore it satisfies the hypothesis of Theorem~\ref{th:mge3}. For this and other results proved in the rest of this section, it will be useful to state a specialization of Proposition~\ref{prop:min} to hypergraph PIN models.

In the case of hypergraph PIN models, for any $B\in\Omega$, $\Delta(\cP_B)$ can be written as $\Delta(\cP_B)=\frac{\sum_{e\in\cE}[P_B(e)-1]}{|\cP_B|-1}$, where $P_B(e)$ is the number of parts of the partition $\cP_B$ intersecting with $e$. On the other hand, $\Delta(\cS)=\frac{(t-1)|\cE|}{m-1}$. Hence, Proposition~\ref{prop:min} can be rewritten for the PIN model as

\begin{corollary}
\label{min:PIN}
For a PIN model described on a $t$-uniform hypergraph, the singleton partition $\cS$ is \\
\emph{(a)} a minimizer for $\mathbf{I}(X_{\cM})$ iff $\frac{(t-1)|\cE|}{m-1} \le \frac{\sum_{e\in\cE}[P_B(e)-1]}{|\cP_B|-1}$ $\forall\,B \in \Omega$; \\
\emph{(b)} the unique minimizer for $\mathbf{I}(X_{\cM})$ iff $\frac{(t-1)|\cE|}{m-1} < \frac{\sum_{e\in\cE}[P_B(e)-1]}{|\cP_B|-1}$ $\forall\,B \in \Omega$.
\end{corollary}

\begin{corollary}
\label{ex:CPIN}
The PIN model on $K_{m,t}$ is strict Type $\cS$.
\end{corollary}

The proof is a relatively straightforward matter of checking that the condition in Corollary~\ref{min:PIN}(b) holds --- see Appendix~\ref{app:typeS} for the details.

We next give an example of a non-isentropic source which is strict Type $\cS$. Consider the PIN model defined on a $k$-regular $k$-edge-connected graph ($t=2$).  Formally, a graph $\cG=(\cV,\cE)$ is called \emph{$k$-regular} if every vertex in $v\in\cV$ has degree $k$, i.e., there are exactly $k$ edges in $\cE$ which are incident with the vertex $v$. A graph is called \emph{$k$-edge-connected} if deletion of any $k$-subset of $\cE$ does not disconnect the graph, but there exists at least one $(k+1)$-subset of $\cE$ the removal of which disconnects the graph.

\begin{corollary}
\label{kkgraph}
A PIN model on any $k$-regular, $k$-edge-connected graph is strict Type $\cS$.
\end{corollary} 

The proof is again an application of Corollary~\ref{min:PIN}(b); the details are in Appendix~\ref{app:typeS}.

The $m$-cycle $\cC_m$ is a special case of a $k$-regular and $k$-edge-connected graph, with $k=2$. Formally, the $m$-cycle $\cC_m=(\cV,\cE)$, is a graph with $\cV=\cM$ and $\cE=\biggl(\bigcup_{i=1}^{m-1}\{\{i,i+1\}\}\biggr)\bigcup\{\{1,m\}\}$. The complete graph $K_{m,2}$ is another example of a $k$-regular $k$-edge-connected graph with $k=m-1$. There is in fact a broad class of $k$-regular $k$-edge-connected graphs called the Harary graphs (see \cite{Har62} and \cite{KMS13}) of which $\cC_m$ and $K_{m,2}$ are special cases. 

So far, the only example we have seen of a strict Type $\cS$ source on a $t$-uniform hypergraph, with $t>2$, has been the PIN model on the complete $t$-uniform hypergraph, $K_{m,t}$. It is natural to ask whether other classes of PIN models on $t$-uniform hypergraphs ($t>2$) exist which are strict Type $\cS$. The answer is `yes'. We will construct a class of uniform hypergraphs with $t=3$, such that the PIN models on them are strict Type $\cS$. To do this, we introduce the \emph{Steiner triple system} (STS) defined on the set $\cM$. An STS on $\cM$ is a collection of 3-subsets of $\cM$, which we will denote by $\STS(\cM)$, such that any pair of elements from $\cM$ is a subset of exactly one element of $\STS(\cM)$. A trivial example of an STS is $m=3$ and $\STS(\cM)=\{\{1,2,3\}\}$. It is a fact that such collections indeed exist as long as $\text{gcd}(m-2,6)=1$ (see \cite[Theorem~2.10]{CR99}). For example, consider $m=7$. Then, $\STS(\cM)=\{\{1,2,4\},\{2,3,5\},\{3,4,6\},\{1,5,6\},\{2,6,7\},\{1,3,7\},\{4,5,7\}\}$. Now, consider the 3-uniform hypergraphs $\cH_{\STS}=(\cM,\STS(\cM))$, for all $m$ such that $\STS(\cM)$ exists. We will show that a PIN model defined on $\cH_{\STS}$ with $m>3$ is strict Type $\cS$. 

\begin{corollary}
\label{STS}
A PIN model on $\cH_{\STS}$ with $m>3$ is strict Type $\cS$.
\end{corollary} 

Again, the proof of the corollary is given in Appendix~\ref{app:typeS}.

Corollaries~\ref{ex:CPIN}, \ref{kkgraph} and \ref{STS} show that the PIN models on $K_{m,t}$, $k$-regular $k$-edge-connected graphs, and $\cH_{\STS}$ satisfy the hypotheses of both Theorems~\ref{th:uh} and \ref{th:mge3}. Thus, for these sources to achieve SK capacity, an omnivocal communication is required. Also, the minimum rate of communication required is $R_{\CO}$. Hence, in terms of public communication, these are the worst-case sources.

\section{Concluding remarks}\label{sec:conc}

This paper dealt with two important aspects of the public communication required to generate maximal-rate SKs in the multiterminal source model, one being the communication complexity $R_{\SK}$, and the other being omnivocality. By extending the arguments in \cite{Tyagi13} to the setting of multiple terminals, we derived a lower bound on $R_{\SK}$ in terms of an information-theoretic quantity called the (multiterminal) interactive common information. In the two-terminal case, it was shown in \cite{Tyagi13} that this bound is always tight, albeit under a weaker notion of secrecy. Proving such a result for the general multiterminal case remains an open problem. 

The minimum rate of communication for omniscience, $R_{\CO}$, is still the best known upper bound on $R_{\SK}$. We proved that uniform hypergraph PIN models satisfying a certain ``Type $\cS$'' condition are $R_{\SK}$-maximal. In other words, for these PIN models, $R_{\SK}$ is equal to $R_{\CO}$. It was also shown via counterexamples that the Type $\cS$ condition is not sufficient to guarantee $R_{\SK}$-maximality for an arbitrary multiterminal source model. A complete characterization of $R_{\SK}$-maximal sources is an interesting open problem.

It should be pointed out that neither our lower bound nor the $R_{\CO}$ upper bound takes into account the fact that the public communication is allowed to be \emph{interactive}. It is possible that incorporating this information somehow leads to better bounds on $R_{\SK}$.

The problem of characterizing communication complexity in the multiterminal source model is the stepping stone towards two bigger problems of interest. One is to characterize the communication rate region required to achieve SK capacity. The second problem is that of determining the minimum rate of communication required to generate an SK of any arbitrary rate less than or equal to SK capacity. Both these questions appear to be difficult to answer at this point. In fact, these questions are still open for the two-terminal case. It should be pointed out that these questions have been answered for a model similar to the multiterminal source model in \cite{LCV15}. However, that model has severe constraints on the eavesdroppers, which makes it somewhat less interesting.

On the issue of omnivocality, we proved that for all strict Type $\cS$ sources, omnivocality is needed to achieve SK capacity. The converse of this fact, i.e., omnivocality is required only if the source is strict Type $\cS$ turns out to be true for three terminals, but no longer holds for four or more terminals. A more general problem along these lines is, given an arbitrary multiterminal source model, what is the minimum number of terminals that must participate in a public communication to generate a maximal-rate SK for the entire set of terminals? The answer to the ``dual" of this problem, i.e., what is the maximum rate of SK that can be generated when a fixed number of terminals remain silent, is already known from the work of Gohari and Anantharam \cite{GA10}. 

\begin{appendices}

\section{Proof of Lemma \ref{lem:mi}}\label{app:proof}

First we state two lemmas which we will require for the proof. 

\begin{lemma}
\label{lem:rv:2}
For independent random variables $X$,$Y$ and $W$, and any other random variable $Z$, we have
$$
I(X;Z|W)\leq I(X;Z|W,Y).
$$
\end{lemma}
\begin{proof}
This follows by expanding $I(X;Y,Z | W)$ in two different ways using the chain rule, and noting that $I(X;Y | W) = 0$.
\end{proof}

\begin{lemma}
\label{lem:rv:1}
For independent random variables $X$ and $Y$, and any other random variable $Z$, we have
$$
I(X;Z)+I(Y;Z)\leq I(X,Y;Z).
$$
\end{lemma}
\begin{proof}
By Lemma~\ref{lem:rv:2}, we have $I(X;Z) \le I(X;Z | Y)$, and hence, $I(X;Z)+I(Y;Z)\leq I(X;Z|Y) + I(Y;Z) = I(X,Y;Z)$.
\end{proof}

\medskip

We first show that it is enough to prove Lemma~\ref{lem:mi} for the complete $t$-uniform hypergraph PIN model $K_{m,t}$ (refer to Section~\ref{sec:singleton} for details on $K_{m,t}$) and the corresponding source $X_{\cM}^n$. Consider any $t$-uniform hypergraph $\cH=(\cV,\cE)$ with $|\cV|\: = m$ and the corresponding source $\hat{X}_{\cM}^n$, and fix a function $\textbf{L}$ of $\hat{X}_{\mathcal{M}}^n$. For any $t$-subset $e$ of $\cV$, define $r(e)$ to be the number of times it occurs in the multiset $\cE$, and call $\displaystyle r=\max_{e\subset\cV:{|e|}=t}r(e)$. Now, construct a new source as follows: To the multiset $\cE^{(n)}$ add $n(r-r(e))$ copies of each $t$-subset $e$ of $\cV$. Associate with each of these newly added subsets independent Ber(1/2) random variables, which are independent of the pre-existing Ber(1/2) random variables as well. Observe that the source thus constructed is none other than $X_{\cM}^{nr}$. Moreover, we clearly have $\sum_{i=1}^m I(X_i^{nr};\textbf{L})\geq \sum_{i=1}^m I(\hat{X}_i^n;\textbf{L})$, and hence it is enough to show  
that $tH(\BL)\geq \sum_{i=1}^m I(X_i^{nr};\textbf{L})$.

For the rest of proof we will take $X_{\cM}^n$ to be the source described on $K_{m,t}$. We also have $I(X_{\cM}^n;\textbf{L})=H(\textbf{L})$ from the fact that $\textbf{L}$ is a function of $X_{\cM}^n$. We now show that the PIN model on $K_{m,t}$ satisfies
\begin{equation}
\sum_{i=1}^m I((\xi_e^n: i\in e, e\in\cE);\textbf{L})\leq t \, I((\xi_e^n: e\in\cE);\textbf{L}), \label{eq:mi1} 
\end{equation}
where $\xi_e^n$ represents the collection of the $n$ $\xi_e$'s associated with the $n$ copies of the hyperedge $e$ in $\cE^{(n)}$.

For any $i\in\cM$, let $\cE_i$ denote the set of hyperedges containing $i$, so that the left-hand side of \eqref{eq:mi1} can be expressed as $\sum_{i=1}^m I\bigl((\xi_e^n:e\in\cE_{i});\textbf{L}\bigr)$. Now, we write $\cE_i$ as a union of two disjoint sets $\cE_{\geq i}$ and $\cE_{\ngtr i}$, i.e., $\cE_i=\cE_{\geq i}\mathop{\dot{\bigcup}}\cE_{\ngtr i}$. The set $\cE_{\geq i}$ is the subset of $\cE_i$ containing no terminals from $\{1,2,\ldots,i-1\}$. The set $\cE_{\ngtr i}$ is thus the subset of  $\cE_i$ containing at least one terminal from $\{1,2,\ldots,i-1\}$. Observe that we have $|\cE_{\geq i}|=\binom{m-i}{t-1}$ for $1\leq i\leq m-t+1$  and $|\cE_{\geq i}|=0$ for $m-t+2\leq i\leq m$. Therefore,
\begin{align}
\sum_{i=1}^m & I\bigl((\xi_e^n:e\in\cE_{i});\textbf{L}\bigr) \nonumber \\
&= I\left(\left(\xi_e^n:e\in\cE_{\geq1}\right);\textbf{L}\right)+ \sum_{i=2}^{m-t+1}\biggl[I\left(\left(\xi_e^n:e\in\cE_{\ngtr i}\right);\BL\right)+I\left(\left(\xi_e^n:e\in\cE_{\geq i}\right);\BL\Big|\left(\xi_e^n:e\in\cE_{\ngtr i}\right)\right)\biggr]\notag\\
&\hspace{1.2em}+\sum_{i=m-t+2}^{m} I\left(\left(\xi_e^n:e\in\cE_{i}\right);\textbf{L}\right) \nonumber \\
& \leq  I\left(\left(\xi_{e}^n:e\in\cE_{\geq1}\right);\textbf{L}\right)+\sum_{i=2}^{m-t+1} I\left(\left(\xi_e^n:e\in\cE_{\geq i}\right);\BL\Big|\biggl(\xi_e^n:e\in\bigcup_{j\leq i}\cE_{\ngtr j}\biggr)\right)+\sum_{i=2}^{m-t+1}I\left(\left(\xi_e^n:e\in\cE_{\ngtr i}\right);\BL\right)\notag\\
&\hspace{1.2em}+\sum_{i=m-t+2}^{m}I\left(\left(\xi_e^n:e\in\cE_{i}\right);\textbf{L}\right) \label{lem:mi:1} \\
& = \underbrace{I\left(\left(\xi_e^n:e\in\cE\right);\BL\right)}_{P}+\underbrace{\sum_{i=2}^{m-t+1}I\left(\left(\xi_e^n:e\in\cE_{\ngtr i}\right);\BL\right)}_{Q}+\underbrace{\sum_{i=m-t+2}^m I\left(\left(\xi_e^n:e\in\cE_{i}\right);\textbf{L}\right)}_{R} \label{lem:mi:2} 
\end{align}
where \eqref{lem:mi:1} follows from Lemma \ref{lem:rv:2}. Note that for $t=2$, \eqref{eq:mi1} follows directly from \eqref{lem:mi:2}:  by virtue of Lemma~\ref{lem:rv:1}, we have $Q + R \le P$, so that the right-hand side (RHS) of \eqref{lem:mi:2} is at most $2P$, as desired.  However, the case of $t>2$ is not as simple and needs further work.

To achieve the RHS of \eqref{eq:mi1}, we require $Q+R \le (t-1)P$. We proceed by defining $Q(i)=I\left(\left(\xi_e^n:e\in\cE_{\ngtr i}\right);\BL\right)$ for all $2\leq i\leq m-t+1$, and thus, $Q=\sum_{i=2}^{m-t+1}Q(i)$. Similarly, define $R(i)=I\left(\left(\xi_e^n: e\in\cE_{i}\right);\textbf{L}\right)$ for all $m-t+2\leq i\leq m$, so that $R=\sum_{i=m-t+2}^m R(i)$. The key ideas are the following: 
\begin{enumerate}
\item Expand each $Q(i)$ using the chain rule into conditional mutual information terms of the form $I(\xi_e^n;\BL|\cdots)$, and further condition them on additional $\xi_{\tilde{e}}^n$s appropriately.
\item Allocate these conditional mutual information terms to appropriate $R(i)$s.
\item Use the chain rule to sum each $R(i)$ and the terms allocated to it to obtain $P$. 
\end{enumerate}
Since the conditional mutual information term $I(\xi_e^n;\BL|\cdots)$ can only increase upon further conditioning on additional $\xi_{\tilde{e}}^n$s (by Lemma \ref{lem:rv:2}), we have $Q+R\leq (t-1)P$ as required.

To proceed, we need to define a total ordering on the set $\cE$. We represent a hyperedge $e$ as a $t$-tuple $(i_1i_2\ldots i_t)$, with the $i_j$s, $1\leq j\leq t$, being the terminals which are contained in $e$, ordered according to $i_1<i_2<\ldots<i_t$. We will use `$<$' to denote the lexicographic ordering of the $t$-tuples (hyperedges) in $\cE$. Furthermore, based on the ordering `$<$', we index the hyperedges of $\cE$ as $e_j$, $1\leq j\leq\binom{m}{t}$, satisfying $e_i<e_j$ iff $i<j$. As an example, Table \ref{tab:order} illustrates the indexing of the hyperedges in $K_{5,3}$.

\begin{table}[ht!]
\caption{Indexing of the hyperedges in $K_{5,3}$}
\label{tab:order}
\begin{center} 
\small 
\begin{tabular}{|c|c|} \hline
Hyperedge & Index\\\hline
$(123)$ & 1 \\\hline
$(124)$ & 2 \\\hline
$(125)$ & 3 \\\hline
$(134)$ & 4 \\\hline
$(135)$ & 5 \\\hline
$(145)$ & 6 \\\hline
$(234)$ & 7 \\\hline
$(235)$ & 8 \\\hline
$(245)$ & 9 \\\hline
$(345)$ & 10 \\\hline
\end{tabular}
\end{center}
\end{table}

To proceed further, using the chain rule we expand each $Q(i)$ into a sum of conditional mutual information terms of the form $Q_e\triangleq I(\xi_e^n;\BL|(\xi_{\tilde{e}}^n:\tilde{e}<e,\tilde{e}\in\cE))$ as follows:
\begin{align}
Q(i)&=I((\xi_e^n:e\in\cE_{\ngtr i});\BL) \notag \\
      &=\sum_{e\in\cE_{\ngtr i}} I(\xi_e^n;\BL|(\xi_{\tilde{e}}^n:\tilde{e}<e,\tilde{e}\in\cE_{\ngtr i})) \notag \\
      &\leq\sum_{e\in\cE_{\ngtr i}} I(\xi_e^n;\BL|(\xi_{\tilde{e}}^n:\tilde{e}<e,\tilde{e}\in\cE)) \label{lem:mi:3} \\
      &=\sum_{e\in\cE_{\ngtr i}}Q_e \label{lem:mi:4}
\end{align}
where \eqref{lem:mi:3} follows from Lemma \ref{lem:rv:2}. Hence, we have $Q\leq\sum_{i=2}^{m-t+1}\sum_{e\in\cE_{\ngtr i}}Q_e$. A total of $\sum_{i=2}^{m-t+1}\biggl[\binom{m-1}{t-1}-\binom{m-i}{t-1}\biggr]=(t-1)\binom{m-1}{t}$ $Q_e$ terms are generated. Next, each $R(i)$ is allocated $\binom{m-1}{t}$ terms $Q_{e_j}$, $1 \le j \le \binom{m}{t}$, satisfying $i \notin e_j$. This allocation procedure is explained in detail below and is also formalized in Algorithm~\ref{alg:ta}. We add a further conditioning on each $Q_{e_j}$ allocated to $R(i)$ to make it $Q_{e_{j|i}}\triangleq I(\xi_{e_j}^n;\BL|(\xi_{\tilde{e}}^n:\tilde{e}<e_j,\tilde{e}\in\cE),(\xi_{\tilde{e}}^n:\tilde{e}\in\cE_{i}))$. Lemma~\ref{lem:rv:2} and the definition of $Q_{e_{j|i}}$ ensure that $R(i)+\sum_{j:i\notin e_j}Q_{e_j}\leq R(i)+\sum_{j:i\notin e_j}Q_{e_{j|i}}=P$.

We now give a more detailed description of the allocation procedure. Construct a table $T$ with rows indexed by $i = 2,3, \ldots, m-t+1$ and the columns indexed by $j = 1,2,\ldots,\binom{m}{t}$. This table records the availability (for allocation) of a $Q_{e_j}$ from the expansion of $Q(i)$ in \eqref{lem:mi:4}. Initialize the table as follows: $T(i,j)=1$ if a $Q_{e_j}$ came from $Q(i)$ in \eqref{lem:mi:4}; else $T(i,j)=0$. We carry out the allocation procedure on each $R(i)$ in ascending order of $i$. The procedure of allocation is as follows. The idea is to allocate the necessary $Q_{e_j}$s to $R(i)$ in ascending order of $j$. Once an $i$ and $e_j$ are fixed, we test whether $i\notin e_j$ is satisfied. If not, we increment $j$ by 1. If  $i\notin e_j$ is satisfied, then the availability of $Q_{e_j}$ from $Q(k)$, for all $2\leq k\leq m-t+1$, is checked using the table $T$. The smallest $k$ which satisfies $T(k,j)=1$ is chosen, and $R(i)$ is allocated the $Q_{e_j}$ coming from that $Q(k)$. The table is then updated with $T(k,j)=0$ to record that the $Q_{e_j}$ from that $Q(k)$ is no longer available for allocation. We then increment $j$ by 1 and repeat the allocation procedure. Once all $Q_{e_j}$s with $i\notin e_j$ have been allocated to $R(i)$, we begin the allocation procedure for $R(i+1)$. We formally summarize this allocation procedure in Algorithm~\ref{alg:ta}.
\begin{algorithm}
\caption{}
\label{alg:ta}
\begin{algorithmic}
\State $i=m-t+2,j=1$.
\While{$i\leq m$}
\If{$i\notin e_j$}
\State $k=2$.
\While{$k\leq m-t+1$}
\If{$T(k,j)=1$}
\State \hspace{-0.5cm}Choose the $Q_{e_j}$ coming from $Q(k)$ in \eqref{lem:mi:4}.
\State \hspace{-0.5cm}Add the additional conditioning to make it $Q_{e_{j|i}}$.
\State \hspace{-0.5cm}Allocate this term to $R(i).$
\State \hspace{-0.5cm}$T(k,j)\gets 0$.
\State \hspace{-0.5cm}Break.
\EndIf
\If{$T(k,j)=0$ \&\& $k=m-t+1$} 
\State Declare ERROR and halt.
\EndIf
\State $k\gets k+1$.
\EndWhile
\EndIf
\State $j\gets j+1$.
\If{$j=\binom{m}{t}+1$}
\State $i\gets i+1$.
\State $j\gets 1$.
\EndIf
\EndWhile
\end{algorithmic}
\end{algorithm}

The flow of Algorithm \ref{alg:ta} for $K_{5,3}$ is illustrated in Example \ref{ex:53} further below. We now make the following claims:
\begin{claim}
\label{cl:1}
Algorithm \ref{alg:ta} never terminates in ERROR.
\end{claim}
\begin{claim}
\label{cl:2}
Algorithm \ref{alg:ta} exhausts all the $Q_e$ terms generated in \eqref{lem:mi:4}.
\end{claim}

Claim~\ref{cl:1} ensures that each $R(i)$, for all $m-t+2\leq i\leq m$, is allocated all the $Q_{e_j}$s satisfying $i\notin e_j$. Therefore, using Claim~\ref{cl:2}, we have 
\begin{align*}
Q+R \ &= \ \sum_{i=m-t+2}^m \left[R(i)+\sum_{j:i\notin e_j}Q_{e_j}\right] \notag\\
& \ \leq\sum_{i=m-t+2}^m \left[R(i)+\sum_{j:i\notin e_j}Q_{e_{j|i}}\right] \ = \ (t-1)P. \notag
\end{align*}
This completes the proof of Lemma~\ref{lem:mi}, modulo the proofs of Claims \ref{cl:1} and \ref{cl:2}, which we give below.

\begin{IEEEproof}[Proof of Claim \ref{cl:1}]
ERROR is possible only if for some $m-t+2\leq i\leq m$ and for some $e$ satisfying $i\notin e$, all the $Q_e$ terms generated in \eqref{lem:mi:4} have already been allocated. This is impossible as there are always enough $Q_e$s. To see this, suppose $e$ contains $t-1-p$ terminals from $\{m-t+2,\ldots,m\}$, i.e., there are $p$ $R(i)$s requiring an allocation of $Q_e$. Since the hypergraph is $t$-uniform, $e$ must contain $p+1$ terminals from $\{1,2,\ldots,m-t+1\}$. This implies that the total number of $Q_e$s generated in \eqref{lem:mi:4} is $p$. Therefore, we clearly have enough $Q_e$s for all $R(i)$s.
\end{IEEEproof}

\begin{IEEEproof}[Proof of Claim \ref{cl:2}]
As discussed earlier, the total number of $Q_e$ terms generated in \eqref{lem:mi:4} is $(t-1)\binom{m-1}{t}$. Also, the total number of $Q_e$ terms required by each $R(i)$ is $\binom{m-1}{t}$. Therefore, using Claim~\ref{cl:1}, the claim follows.
\end{IEEEproof}

\begin{example}
We illustrate how Algorithm~\ref{alg:ta} proceeds for $K_{5,3}$. Denote the hyperedges in $\cE$ using $3$-tuples, i.e., the hyperedge containing terminals $1$, $2$ and $3$ is $(123)$. The indexing of $\cE$ is illustrated in Table~\ref{tab:order}. So for this case we have $Q(2)=I(\xi_{(123)}^n,\xi_{(124)}^n,\xi_{(125)}^n;\BL)$ and $Q(3)=I(\xi_{(123)}^n,\xi_{(134)}^n,\xi_{(135)}^n,\xi_{(234)}^n,\xi_{(235)}^n;\BL)$. Thus, \eqref{lem:mi:4} takes the form
\begin{align}
Q(2) & \leq I(\xi_{(123)}^n;\BL)+I(\xi_{(124)}^n;\BL|(\xi_e^n:e<(124))+I(\xi_{(125)}^n;\BL|(\xi_e^n:e<(125)) \label{ex:1}\\
Q(3) & \leq I(\xi_{(123)}^n;\BL)+I(\xi_{(134)}^n;\BL|(\xi_e^n:e<(134))+I(\xi_{(135)}^n;\BL|(\xi_e^n:e<(135))\notag\\
        &\hspace{1em}+I(\xi_{(234)}^n;\BL|(\xi_e^n:e<(234))+I(\xi_{(235)}^n;\BL|(\xi_e^n:e<(235)) \label{ex:2}
\end{align}
Observe that $R(4)$ and $R(5)$ require four $Q_e$ terms each, and a total of eight $Q_e$ terms are in fact available from \eqref{ex:1} and \eqref{ex:2}. The table $T$ is initialized as follows:
\begin{center} 
\small 
\begin{tabular}{|c||c|c|c|c|c|c|c|c|c|c|} \hline
& 1 & 2 &3 & 4 & 5 & 6 & 7 & 8 & 9 & 10\\\hline\hline
2& 1&1&1&0&0&0&0&0&0&0\\\hline
3& 1&0&0&1&1&0&1&1&0&0\\\hline
\end{tabular}
\end{center}

We will now illustrate a few of the allocations carried out by Algorithm \ref{alg:ta}. The algorithm begins with $i=4$ and $j=1$ and $Q_{(123)}$ needs to be allocated to $R(4)$. With $k=2$ we see that $T(k,1)=1$, and hence we allocate $Q_{(123)}$ coming from $Q(2)$ to $R(4)$. The table $T$ is then updated as below.
\begin{center} 
\small 
\begin{tabular}{|c||c|c|c|c|c|c|c|c|c|c|} \hline
& 1 & 2 &3 & 4 & 5 & 6 & 7 & 8 & 9 & 10\\\hline \hline
2& 0&1&1&0&0&0&0&0&0&0\\\hline
3& 1&0&0&1&1&0&1&1&0&0\\\hline
\end{tabular}
\end{center}

Next we will illustrate the allocation of $Q_{(123)}$ to $R(5)$, i.e., $i=5$ and $j=1$. The state of the table $T$ just before this step is shown below.
\begin{center} 
\small 
\begin{tabular}{|c||c|c|c|c|c|c|c|c|c|c|} \hline
& 1 & 2 &3 & 4 & 5 & 6 & 7 & 8 & 9 & 10\\\hline \hline
2& 0&1&0&0&0&0&0&0&0&0\\\hline
3& 1&0&0&1&0&0&1&0&0&0\\\hline
\end{tabular}
\end{center}

Setting $k=2$, we see that $T(k,1)=0$. So, we move to $k=3$, for which $T(k,1)=1$. Hence the $Q_{(123)}$ term coming from $Q(3)$ is allocated to $R(5)$, and the table $T$ is updated as below.
\begin{center} 
\small 
\begin{tabular}{|c||c|c|c|c|c|c|c|c|c|c|} \hline
& 1 & 2 &3 & 4 & 5 & 6 & 7 & 8 & 9 & 10\\\hline \hline
2& 0&1&0&0&0&0&0&0&0&0\\\hline
3& 0&0&0&1&0&0&1&0&0&0\\\hline
\end{tabular}
\end{center}

We give one last example of an allocation. Observe that $e=(234)$ is the largest (in terms of the ordering on $\cE$) hyperedge such that $Q_e$ needs to be allocated to $R(5)$. We will now illustrate this step. This happens when $i=5$ and $j=7$. The updated table $T$ just before this step is shown below.

\begin{center} 
\small 
\begin{tabular}{|c||c|c|c|c|c|c|c|c|c|c|} \hline
& 1 & 2 &3 & 4 & 5 & 6 & 7 & 8 & 9 & 10\\\hline \hline
2& 0&0&0&0&0&0&0&0&0&0\\\hline
3& 0&0&0&0&0&0&1&0&0&0\\\hline
\end{tabular}
\end{center}

With $k=2$, we see that $T(k,7)=0$. So set $k=3$, and note that $T(k,7)=1$. So, we allocate to $R(5)$ the $Q_{(234)}$ term contributed by $Q(3)$. Upon updating, the table $T$ now has all entries to be $0$. Observe that at this point no other allocation is required, as the $Q_{e_j}$s for $j=8$, $9$ and $10$ are not required by $R(5)$ since terminal $5$ is contained in each of $e_8$, $e_9$ and $e_{10}$. Thus Algorithm \ref{alg:ta} successfully terminates. Finally, we rewrite \eqref{ex:1} and \eqref{ex:2} with underbraces showing the $R(i)$ term to which each $Q_e$ term was allocated by Algorithm \ref{alg:ta}.

\begin{align}
Q(2) & \leq \underbrace{I(\xi_{(123)}^n;\BL)}_{R(4)}+\underbrace{I(\xi_{(124)}^n;\BL|(\xi_e^n:e<(124))}_{R(5)}+\underbrace{I(\xi_{(125)}^n;\BL|(\xi_e^n:e<(125))}_{R(4)} \label{ex:53:1}\\
Q(3) & \leq \underbrace{I(\xi_{(123)}^n;\BL)}_{R(5)}+\underbrace{I(\xi_{(134)}^n;\BL|(\xi_e^n:e<(134))}_{R(5)}+\underbrace{I(\xi_{(135)}^n;\BL|(\xi_e^n:e<(135))}_{R(4)}\notag\\
        &\hspace{1em}+\underbrace{I(\xi_{(234)}^n;\BL|(\xi_e^n:e<(234))}_{R(5)}+\underbrace{I(\xi_{(235)}^n;\BL|(\xi_e^n:e<(235))}_{R(4)} \label{ex:53:2}
\end{align}

It can be clearly seen from \eqref{ex:53:1} and \eqref{ex:53:2} that $R(i), i=4,5,$ have each been allocated with all $Q_e$s with $i\notin e$, and no $Q_e$ is left unallocated.

\label{ex:53}
\end{example}

\section{The proof of Lemma~\ref{lem:tree}}\label{app:tree}

Fix an $n\in\mathbb{N}$ and let $\{T_1,T_2,\ldots,T_{\sigma^{(n)}}\}$ be a set of edge-disjoint spanning trees of $\cG^{(n)}$ of maximum cardinality $\sigma^{(n)} :=\sigma(\cG^{(n)})$. We will run Protocol~1 of \cite{TKSV12} independently on each of the trees $T_j, 1\leq j\leq\sigma^{(n)}$. For the sake of completeness, we describe the protocol below. 

Fix a spanning tree $T_j, 1\leq j\leq\sigma^{(n)}$, and fix a specific edge $e$ from the set of edges of $T_j$. Define $\xi(T_j) := \xi_e$, where, as usual, $\xi_e$ denotes the random variable associated with the edge $e$. For any vertex $i\in\cM$, denote by $d_j(i)$ the degree of the vertex $i$ in the spanning tree $T_j$. For any vertex $i$ satisfying $d_j(i)>1$, without loss of generality we label the edges of $T_j$ incident on it by $e(1),e(2),\ldots,e(d)$, where $d = d_j(i)$. The communication from terminal $i$ derived from $T_j$ is $\textbf{F}_{T_j}(i) := \bigl(\xi_{e(1)}\oplus \xi_{e(2)},\xi_{e(2)}\oplus \xi_{e(3)},\ldots,\xi_{e(d-1)}\oplus \xi_{e(d)}\bigr)$, where $\oplus$ denotes the modulo-2 sum. Let $\textbf{F}_{T_j} = \left(\textbf{F}_{T_j}(1), \textbf{F}_{T_j}(2), \ldots, \textbf{F}_{T_j}(m)\right)$, and let $\mathcal{F}_{T_j}$ denote the range of $\textbf{F}_{T_j}$. It is not hard to check the following facts: Firstly, every terminal can recover $\xi(T_j)$ from $\textbf{F}_{T_j}$. Secondly, $I(\textbf{F}_{T_j};\xi(T_j))=0$. Thirdly, 
\begin{gather}
\displaystyle\log|\mathcal{F}_{T_j}|=\sum_{i=1}^m[d_j(i)-1]= \sum_{i=1}^md_j(i)-m=2(m-1)-m=m-2, \label{protorate} 
\end{gather}
where we have used the fact that the number of edges in any spanning tree is $m-1$.

To complete the proof, we show that this protocol has communication rate $(m-2)\overline{\sigma}(\cG)$ and achieves SK capacity. Denote the entire communication $\bigl(\BF_{T_1},\BF_{T_2},\ldots,\BF_{T_{\sigma^{(n)}}}\bigr)$ by $\BF$ and denote its range by $\mathcal{F}$. Set $\BK=\bigl(\xi(T_1),\xi(T_2),\ldots,\allowbreak \xi(T_{\sigma^{(n)}})\bigr)$. Noting that the spanning trees $T_j, 1\leq j\leq\sigma^{(n)} $ are edge-disjoint, we have, using the independence of the random variables associated with the edges in $\cE^{(n)}$, $H(\BK)=\sigma^{(n)}$, $\log|\mathcal{F}|=(m-2)\sigma^{(n)}$ and $I(\BK;\BF)=0$. Therefore, $\BK$ is a secret key satisfying $\displaystyle\lim_{n\to\infty}\frac{1}{n}H(\BK)=\overline{\sigma}(\cG)$, and hence the protocol is capacity-achieving. The protocol used a communication rate of $(m-2)\overline{\sigma}(\cG)$ and thus $R_{\SK}\leq(m-2)\overline{\sigma}(\cG)$.

\section{An example of a non-$R_{\SK}$-maximal strict Type $\cS$ source} \label{app:exunique}

In this section we provide an example of a source which is strict Type $\cS$ and yet is non $R_{\SK}$-maximal. To construct such a source we need to define ``clubbing together" of independent multiterminal sources on $\cM$. Formally, for independent sources $X_{\cM}^n$ and $Y_{\cM}^n$, define the \emph{clubbed} source $Z_{\cM}^n$ as $Z_i^n=(X_i^n,Y_i^n)$, for all $i\in\cM$. $\Pi_X^*$ and $\Pi_Y^*$ are defined to be the sets of partitions of $\cM$ which are minimizers of \eqref{eq:I} for $X_{\cM}^n$ and $Y_{\cM}^n$, respectively. We will denote the communication complexity (resp. minimum rate of communication for omniscience) for the individual sources $X_{\cM}^n$ and $Y_{\cM}^n$ by $R_{\SK_X}$ and $R_{\SK_Y}$ (resp. $R_{\CO_X}$ and $R_{\CO_Y}$) respectively. The clubbed source satisfies the following result.
\begin{proposition}
Consider two independent multiterminal sources $X_{\cM}^n$ and $Y_{\cM}^n$ and the corresponding clubbed source $Z_{\cM}^n$. Then we have
\begin{equation}
\textbf{I}(Z_{\cM})\geq\textbf{I}(X_{\cM})+\textbf{I}(Y_{\cM}) \label{th:club:1}
\end{equation}
with equality iff $\Pi_X^*\bigcap\Pi_Y^*\neq\emptyset$.
\label{prop:club}
\end{proposition}

\begin{IEEEproof}
Consider any partition $\cP=\{A_1,A_2,\cdots,A_{\ell}\}$ of $\cM$. We have
\begin{align}
\Delta(\cP)&=\frac{1}{\ell-1}\left[\sum_{i=1}^{\ell}H(Z_{A_i})-H(Z_{\cM})\right] \notag\\
                 &=\underbrace{\frac{1}{\ell-1}\left[\sum_{i=1}^{\ell}H(X_{A_i})-H(X_{\cM})\right]}_{\Delta_X(\cP)}+\underbrace{\frac{1}{\ell-1}\left[\sum_{i=1}^{\ell}H(Y_{A_i})-H(Y_{\cM})\right]}_{\Delta_Y(\cP)} \label{club:1} 
\end{align}
where \eqref{club:1} follows from the independence of $X_{\cM}^n$ and $Y_{\cM}^n$. 

Thus we have from \eqref{club:1} that $\min_{\cP}\Delta(\cP)\geq\min_{\cP}\Delta_X(\cP)+\min_{\cP}\Delta_Y(\cP)$ with equality iff $\cP\in\Pi_X^*\bigcap\Pi_Y^*$. The result follows.
\end{IEEEproof}

We conclude the section by constructing a non $R_{\SK}$-maximal source with $\cS$ being the unique minimizer in \eqref{eq:I}.

\begin{example}
Consider a clubbed source $Z_{\cM}^n=(X_{\cM}^n,Y_{\cM}^n)$, where $X_{\cM}^n$ is the source described in Example \ref{ex:omni} and $Y_{\cM}^n$ corresponds to the PIN model on a $k$-regular, $k$-edge-connected graph. By Corollary~\ref{kkgraph}, we have $\Pi_Y^*=\{\cS\}$. 

Since $\Pi^*_X\bigcap\Pi^*_Y=\{\cS\}$, using Proposition~\ref{prop:club} we have SK capacity $\textbf{I}(Z_{\cM})=\textbf{I}(X_{\cM})+\textbf{I}(Y_{\cM})$. By independently running protocols achieving $R_{\SK_X}$ and $R_{\SK_Y}$, an SK of rate $\textbf{I}(X_{\cM})+\textbf{I}(Y_{\cM})$, i.e., SK capacity can be achieved. The communication rate used in independently running the two protocols is $R_{\SK_X}+R_{\SK_Y}$. Now, \eqref{omni} and the independence of $X_{\cM}^n$ and $Y_{\cM}^n$ show that $R_{\CO}=R_{\CO_X}+R_{\CO_Y}$. On the other hand, it is shown in Example~\ref{ex:omni} that $R_{\SK_X}<R_{\CO_X}$. Therefore, we have 
\begin{gather}
R_{\SK}\leq R_{\SK_X}+R_{\SK_Y} < R_{\CO_X}+R_{\CO_Y}=R_{\CO}. \notag
\end{gather}
\label{ex:unique}
\end{example}

\section{A non-strict Type $\cS$ source requiring omnivocality\protect\footnote{This example is a contribution of Chan et al. See \cite{Chan14}.}}\label{app:chan}

For $m\geq 4$, consider the multigraph $\cG=(\cV,\cE)$ with $\cV=\cM$ as usual. The multiset $\cE$ consists of $m-2$ copies of the edges $\{i,i+1\}$ for $1\leq i\leq m-1$, and $m-1$ copies of the edge $\{1,m\}$. Using techniques derived in \cite[Theorem 5]{NN10}, it can be shown that for the PIN model defined on $\cG$, we have $\Ixm=m-1$. We will show below that this PIN model is non-strict Type $\cS$, and yet it requires omnivocality to achieve SK capacity. 

We first show that the source is not strict Type $\cS$. Simple computations reveal the following facts: $H(X_i)=2(m-2)$, for all $i\in\{2,3,\ldots,m-1\}$, $H(X_1)=H(X_m)=2(m-2)+1$, $H(X_1,X_m)=3(m-2)+1$ and $H(X_{\cM})=m(m-2)+1$. Using these it is easy to check that $\Delta(\cS)=m-1$, and moreover, $\Delta(\cP^*)=m-1$, where $\cP^*=\{\{1,m\},\{2\},\{3\},\ldots,\{m-1\}\}$. Hence the source $X_{\cM}^n$ is Type $\cS$, but not strict Type $\cS$.

Now, we show that this source requires omnivocality to achieve SK capacity. As in the proof of Theorem~\ref{th:mge3}, we make use of Theorem~\ref{th:silent}, and show that for any $T\subset\cM$ with ${|T|}=m-1$, we have $\Itxm<\Ixm$. Let $T=\cM\setminus\{u\}$ with $u\in\cM$. Using symmetry it is enough to show $\Itxm<\Ixm$ for the following two cases:

Case I: $u=1$.

Case II: $u\in\{2,3,\ldots,m-2\}$.

\medskip

In both cases we will derive lower bounds on $\RTmin$ and hence obtain an upper bound on $\Itxm$. First we deal with Case~I with $T=\{2,3,\ldots,m\}$. In this case, $H(X_T)=m(m-2)+1$. Also, any point in $\cR_T$ satisfies the following constraints from \eqref{eq:lem}:

\begin{align}
\sum_{i=2}^{m-1}R_i & \geq H(X_2,X_3,\ldots,X_{m-1}|X_m)= (m-2)^2 \notag\\
R_m & \geq H(X_m|X_2,X_3,\ldots,X_{m-1})=m-1\notag
\end{align}
Using the above constraints, we have $\RTmin\geq (m-1)+(m-2)^2$. Thus,
\begin{gather}
\Itxm=H(X_T)-\RTmin\leq m(m-2)+1-(m-1)-(m-2)^2=m-2<m-1=\Ixm. \label{app:ex:1}
\end{gather}
Hence, SK capacity cannot be achieved with terminal 1 remaining silent.

Next we deal with Case~II. Assume an arbitrary $u\in\{2,3,\ldots,m-2\}$ is silent. As in Case~I, we have $H(X_T)=m(m-2)+1$. We see from \eqref{eq:lem} that the rate region $\cR_T$ is defined in part by the following constraints:

\begin{align}
\sum_{i=1}^{u-1}R_i + R_m & \geq H(X_1,X_2,\ldots,X_{u-1},X_m|X_{u+1},X_{u+2},\ldots,X_{m-1})= u(m-2)+1 \notag\\
\sum_{i=u+1}^{m-1} R_i& \geq H(X_{u+1},X_{u+2},\ldots,X_{m-1}|X_1,X_2,\ldots,X_{u-1},X_m)=(m-u-1)(m-2)\notag
\end{align}
The above constraints imply that $\RTmin\geq (m-2)(m-1)+1=(m-2)^2+(m-1)$. Hence, as in \eqref{app:ex:1}, we have $\Itxm<\Ixm$.

Therefore, the source $X_{\cM}^n$ cannot attain SK capacity without using omnivocality.

\section{Proofs of Corollaries of Proposition~\ref{prop:min}\label{app:typeS}}

In this section, we give the proofs of Corollaries~\ref{isentropic}, \ref{ex:CPIN}, \ref{kkgraph} and \ref{STS}. We start with the corollary stating that isentropic random variables form a Type $\cS$ source.

\begin{IEEEproof}[Proof of Corollary~\ref{isentropic}]
For a partition $\cP$ of $\cM$ with $|\cP| \, \ge 2$, let us define
$$
\delta(\cP) \ \triangleq \ \frac{1}{|\cP|-1}\sum_{A\in \cP} H(X_{A^c}|X_A) \ = \ H(X_{\cM})-\Delta(\cP). 
$$
By virtue of Proposition~\ref{prop:min}(a), we need to show that $\delta(\cP_B) \le \delta(\cS)$ for all $B \in \Omega$.

For isentropic random variables, the quantity $H(X_B|X_{B^c})$, for any $B \subseteq \cM$, depends only on the cardinality of $B$. Thus, for $1 \le k \le m$, define $g(k) \triangleq H(X_{\{1,2,\ldots,k\}} | X_{\cM\setminus \{1,2,\ldots,k\}})$; also, set $g(0) = 0$. With this, we can write
\begin{align*}
\delta(\cP_B) &= \frac{1}{|B|} \left[H(X_B|X_{B^c}) + \sum_{i \in B} H(X_{\cM\setminus\{i\}}|X_i)\right] \\
& = \frac{1}{|B|} g(|B|) + g(m-1)
\end{align*}
Also, note that $\delta(\cS) = \frac{m}{m-1}g(m-1)$. Thus, we have to show that $\frac{g(|B|)}{|B|} \le \frac{g(m-1)}{m-1}$ for all $B \in \Omega$. We accomplish this by proving that for isentropic random variables, the function $g(k)/k$ is non-decreasing in $k$, or equivalently, $kg(k+1) - (k+1)g(k)$ is always non-negative. Indeed, we have $g(k+1) = H(X_{\cM}) - H(X_{\{k+2,\ldots,m\}})$ and $g(k) = H(X_{\cM}) - H(X_{\{k+1,\ldots,m\}}) = g(k+1) - H(X_{k+1}|X_{\{k+2,\ldots,m\}})$. Thus, 
\begin{equation}
kg(k+1)- (k+1)g(k)= (k+1) \, H(X_{k+1}|X_{\{k+2,\ldots,m\}}) - g(k+1). \notag
\end{equation}
It is straightforward to show that the above quantity is non-negative:
\begin{align*}
g(k+1) &= H(X_{\{1,2,\ldots,k+1\}}|X_{\{k+2,\ldots,m\}}) \\
& \le \sum_{i=1}^{k+1} H(X_{i}|X_{\{k+2,\ldots,m\}}) \\
& = (k+1) H(X_{k+1}|X_{\{k+2,\ldots,m\}}),
\end{align*}
since, for $1 \le i \le k+1$, $H(X_{i}|X_{\{k+2,\ldots,m\}}) = H(X_{k+1}|X_{\{k+2,\ldots,m\}})$ by isentropy.
\end{IEEEproof}

Next, we prove Corollary~\ref{ex:CPIN}, which states that the PIN model on $K_{m,t}$ is strict Type $\cS$.

\begin{IEEEproof}[Proof of Corollary~\ref{ex:CPIN}]
Fix a set $B\subsetneq\cM$ with $|B|\:\leq m-2$. We will use Corollary~\ref{min:PIN} to show that the PIN model on $K_{m,t}$ is strict Type $\cS$. First we make the observation that ${|\cE|}=\binom{m}{t}$ for the case of $K_{m,t}$. To proceed, we need to evaluate the expression $\sum_{e\in\cE}[P_B(e)-1]$. We first consider the case when ${|B|}\geq t$. The fact that $|B|$ is at least $t$ implies that there are $\binom{|B|}{t}$ hyperedges which contain only elements of $B$, i.e., intersect the partition $\cP_B$ in $t$ parts. Now fix an $i$ with $1\leq i\leq t-1$. There are $\binom{|B|}{i}\binom{m-|B|}{t-i}$ hyperedges containing any $i$ terminals from $B$ and any $t-i$ terminals from $\cM\setminus B$, i.e., intersecting the partition $\cP_B$ in $(i+1)$ parts. Any remaining hyperedge will contain terminals from $B^c$ only and hence will intersect the partition $\cP_B$ in only one part. As a result, we have $\sum_{e\in\cE}[P_B(e)-1]=(t-1)\binom{|B|}{t}+\sum_{i=1}^{t-1}\binom{|B|}{i}\binom{m-|B|}{t-i}i=(t-1)\binom{|B|}{t}+|B|\sum_{i=1}^{t-1}\binom{|B|-1}{i-1}\binom{m-|B|}{t-i}$. Observe that $\sum_{i=1}^{t-1}\binom{|B|-1}{i-1}\binom{m-|B|}{t-i}$ is equal to $\binom{|B|-1}{t-1}$ subtracted from the coefficient of  $x^{t-1}$ in the expansion of $(1+x)^{|B|-1}(1+x)^{m-|B|}=(1+x)^{m-1}$. Therefore, $\sum_{i=1}^{t-1}\binom{|B|-1}{i-1}\binom{m-|B|}{t-i}=\binom{m-1}{t-1}-\binom{|B|-1}{t-1}$, and hence, for ${|B|}\geq t$, we have
\begin{equation}
\sum_{e\in\cE}[P_B(e)-1]=|B|\binom{m-1}{t-1}+(t-1)\binom{|B|}{t}-|B|\binom{|B|-1}{t-1}=|B|\binom{m-1}{t-1}-\binom{|B|}{t}. 
\label{eq1:CPIN}
\end{equation}

Next, we turn our attention to the case of $|B|\:<t$. In this case there are no hyperedges containing only terminals in $B$. For any $i$ satisfying $1\leq i\leq|B|$, there exist $\binom{|B|}{i}\binom{m-|B|}{t-i}$ hyperedges intersecting the partition in $(i+1)$ parts, as in the earlier case. However, all the remaining hyperedges are contained in $B^c$ only, and hence play no part in the expression $\sum_{e\in\cE}[P_B(e)-1]$. Thus, noting $|B|\:<t$, we have as in the previous case,
\begin{equation}
\sum_{e\in\cE}[P_B(e)-1]=|B|\sum_{i=1}^{|B|}\binom{|B|-1}{i-1}\binom{m-|B|}{t-i}=|B|\binom{m-1}{t-1}. \label{CPIN2}
\end{equation}

We will now apply Corollary~\ref{min:PIN}. When ${|B|}\geq t$, using \eqref{eq1:CPIN} we have
\begin{align}
\frac{1}{|B|}\sum_{e\in\cE}[P_B(e)-1]-\frac{(t-1)|\cE|}{m-1} & = \binom{m-1}{t-1}-\frac{1}{|B|}\binom{|B|}{t}-\frac{t-1}{m-1}\binom{m}{t} \label{CPIN3}\\
                                      & = \frac{1}{t}\biggl[\frac{(m-1)!\ t}{(m-t)!\ (t-1)!}-\frac{m!}{(t-2)!\ (m-t)!\ (m-1)}-\binom{|B|-1}{t-1}\biggr] \nonumber\\
                                      & = \frac{1}{t}\biggl[\frac{(m-1)!}{(t-2)!\ (m-t)!}\left(\frac{t}{t-1}-\frac{m}{m-1}\right)-\binom{|B|-1}{t-1}\biggr] \nonumber\\
                                      & = \frac{1}{t}\left[\binom{m-2}{t-1}-\binom{|B|-1}{t-1}\right] \notag \\
                                     & > 0 \label{CPIN4}
\end{align}
where \eqref{CPIN4} holds as ${|B|}\leq m-2$. For the case of ${|B|}< t$, we have 
\begin{align}
\frac{1}{|B|}\sum_{e\in\cE}[P_B(e)-1]-\frac{(t-1)|\cE|}{m-1} & = \binom{m-1}{t-1}-\frac{t-1}{m-1}\binom{m}{t}\notag\\
                                                                                                       &=\frac{1}{t}\biggl[\binom{m-2}{t-1}\biggr] \label{CPIN5}\\
                                                                                                      &> 0 \notag
\end{align}
where \eqref{CPIN5} follows from \eqref{CPIN3} and \eqref{CPIN4}. Thus, using Corollary~\ref{min:PIN} we have the result.
\end{IEEEproof}

Next up is the proof of Corollary~\ref{kkgraph}, which states that PIN models on $k$-regular, $k$-edge-connected graphs are strict Type $\cS$.

\begin{IEEEproof}[Proof of Corollary~\ref{kkgraph}]
Consider a $k$-regular, $k$-edge-connected graph  $\cG=(\cV,\cE)$. Using $k$-regularity, we have $|\cE|\;=\frac{km}{2}$. As usual, we fix a $B\subsetneq\cM$ satisfying $1\leq|B|\;\leq m-2$ and proceed to evaluate the expression $\sum_{e\in\cE}[P_B(e)-1]$. Observe that for an ordinary graph, the sum $\sum_{e\in\cE}[P_B(e)-1]=|\cE_{\cP_B}|$, where $\cE_{\cP_B}$ is the set of edges whose end-points lie in different cells of the partition $\cP_B$. To proceed, we perform a graph contraction operation along the partition $\cP_B$ on $\cG$ to get a new graph $\cG'=(\cV',\cE')$. More precisely, we take $\cV'=B\bigcup\{B^c\}$ and $\cE'=\biggl\{\{i,j\}\in\cE: i,j\in B\biggr\}\bigcup\biggl\{\{B^c,i\}: \exists\{i,j\}\in\cE, i\in B, j\in B^c\biggr\}$, so that $|\cE_{\cP_B}|\;=|\cE'|$. Now, the degree of every $v\in\cV'$ satisfying $v\in B$ is $k$, whereas by the $k$-edge connectivity the degree of $B^c$ in $\cG'$ is at least $k$. Hence, we have $|\cE_{\cP_B}|\;=|\cE'|\;\geq \frac{k(|B|+1)}{2}$. Therefore,
\begin{align}
\frac{1}{|B|}\sum_{e\in\cE}[P_B(e)-1]-\frac{|\cE|}{m-1}&=\frac{1}{|B|}|\cE_{\cP_B}|-\frac{km}{2(m-1)} \notag\\
                                                                                                     &\geq \frac{k}{2}\biggl[\frac{|B|+1}{|B|}-\frac{m}{m-1}\biggr] \notag\\
                                                                                                     &>0 \label{kkgraph:1}
\end{align}
where, \eqref{kkgraph:1} follows from the fact that $|B|\;\leq m-2$. Using Corollary~\ref{min:PIN} we have the result.
\end{IEEEproof}

Finally, we give the proof of Corollary~\ref{STS}, which states that a PIN model obtained from a Steiner triple system (STS) is strict Type $\cS$. Recall that $\cH_{\STS} = (\cM,\STS(\cM))$ is a $3$-uniform hypergraph obtained from an STS on $\cM$.

\begin{IEEEproof}[Proof of Corollary~\ref{STS}]
We will use Proposition~\ref{prop:min} to get the result. First, we calculate $H(X_i)$ for any $i\in\cM$. Observe that $H(X_i)$ counts the number of elements of $\STS(\cM)$ containing $i$. Now, fixing $i\in\cM$, there are $m-1$ pairs of elements from $\cM$ which contain $i$. Any set in $\STS(\cM)$ containing $i$ contains two such pairs. Further, by the definition of STS, we know that any such pair is a subset of exactly one element of $\STS(\cM)$. Hence, we have $H(X_i)=\frac{m-1}{2}$. Next, we evaluate $H(X_{\cM})=|\STS(\cM)|$. Note that there are $\binom{m}{2}$ pairs of elements in $\cM$, each pair being a subset of exactly one element of $\STS(\cM)$. Also, each element of $\STS(\cM)$ contains three such pairs. Therefore, we have $H(X_{\cM})=|\STS(\cM)|=\frac{m(m-1)}{6}$. Using these facts, we have $\Delta(\cS)=\frac{1}{m-1}\biggl[\frac{m(m-1)}{2}-\frac{m(m-1)}{6}\biggr]=\frac{m}{3}$. 

Now, fix a $B\subsetneq\cM$ with $1\leq |B|\;\leq m-2$ and evaluate $\Delta(\cP_B)$. We consider two cases: $1\leq |B|\;\leq m-3$ and $|B|\;=m-2$. First, consider $1\leq |B|\;\leq m-3$. To proceed, we calculate a lower bound on $H(X_A)$ for any $A\subsetneq\cM$. Observe that $H(X_A)$ counts the number of sets in $\STS(\cM)$ which contain at least one element from $A$. We will calculate an upper bound on the number of elements of $\STS(\cM)$ containing only elements of $A^c$, and subtract it from $|\STS(\cM)|$ to get the required lower bound. The total number of pairs formed by the elements of $A^c$ is $\binom{m-|A|}{2}$. Again, as each element of $\STS(\cM)$ contains 3 pairs, the required upper bound is $\lfloor\frac{(m-|A|)(m-|A|-1)}{6}\rfloor$. Thus, we have $H(X_A)\geq |\STS(\cM)|-\frac{(m-|A|)(m-|A|-1)}{6}$. So, $H(X_{B^c})\geq |\STS(\cM)|-\frac{|B|(|B|-1)}{6}$, and hence, $\Delta(\cP_B)\geq \frac{1}{|B|}\biggl[\frac{|B|(m-1)}{2}-\frac{|B|(|B|-1)}{6}\biggr]=\frac{m-1}{2}-\frac{|B|-1}{6}$. Therefore, 
\begin{align}
\Delta(\cP_B)-\Delta(\cS)&\geq \frac{m-1}{2}-\frac{|B|-1}{6}-\frac{m}{3} \notag\\
                                        &=\frac{1}{6}[m-2-|B|]\notag\\
                                        &>0 \label{STS:1}
\end{align}
where, \eqref{STS:1} follows from the fact that $|B|\;<m-2$. 

To complete the proof, we show that $\Delta(\cP_B)-\Delta(\cS)>0$ is satisfied when $|B|\;=m-2$. To this end, we fix a $B=\cM\setminus\{i,j\}$, where $i,j\in\cM$. We will exactly calculate $H(X_{B^c})$, which is the number of elements of $\STS(\cM)$ containing at least one of $i$ and $j$. It has been shown earlier that $i$ and $j$ each occur in exactly $\frac{m-1}{2}$ elements, and they occur together exactly once. Therefore, we have $H(X_{B^c})=m-2$, and hence, $\Delta(\cP_B)=\frac{1}{m-2}\biggl[\frac{(m-2)(m-1)}{2}+(m-2)-\frac{m(m-1)}{6}\biggr]$. Thus,
\begin{align}
\Delta(\cP_B)-\Delta(\cS)&=\frac{1}{m-2}\biggl[\frac{(m-2)(m-1)}{2}+(m-2)-\frac{m(m-1)}{6}\biggr]-\frac{m}{3}\notag\\
                                        &=\frac{m-3}{3(m-2)}\notag\\
                                        &>0 \label{STS:2}
\end{align} 
where \eqref{STS:2} follows from the fact that $m>3$.
\end{IEEEproof}

\end{appendices}


\begin{thebibliography}{99}

\bibitem{Maurer93} 
U. M.~Maurer, ``Secret key agreement by public discussion from common information,'' \emph{IEEE Trans.\ Inf.\ Theory}, vol.\ 39, pp.\ 733--742, May 1993.

\bibitem{AC93} 
R.~Ahlswede and I.~Csisz{\'a}r, ``Common randomness in information theory and cryptography, part I: Secret sharing,'' \emph{IEEE Trans.\ Inf.\ Theory}, vol.\ 39, pp.\ 1121--1132, July 1993.

\bibitem{CN04}
I.~Csisz{\'a}r and P.~Narayan, ``Secrecy capacities for multiple terminals,'' \emph{IEEE Trans.\ Inf.\ Theory}, vol.\ 50, pp.\ 3047--3061, Dec.\ 2004.

\bibitem{CN08}
I.~Csisz{\'a}r and P.~Narayan, ``Secrecy capacities for multiterminal channel models,'' \emph{IEEE Trans.\ Inf.\ Theory}, vol.\ 54, no.\ 6, pp.\ 2437--2452, June \ 2008.

\bibitem{NYBNR10}
S.~Nitinawarat, C.~Ye, A.~Barg, P.~Narayan and A.~Reznik, ``Secret key generation for a pairwise independent network model,'' \emph{IEEE Trans.\ Inf.\ Theory}, vol.\ 56, pp.\ 6482--6489, Dec.\ 2010.

\bibitem{NN10}
S.~Nitinawarat and P.~Narayan, ``Perfect omniscience, perfect secrecy and Steiner tree packing,'' \emph{IEEE Trans.\ Inf.\ Theory}, vol.\ 56, no.\ 12, pp.\ 6490--6500, Dec.\ 2010. 

\bibitem{GA10}
A. A.~Gohari and V.~Anantharam, ``Information-theoretic key agreement of multiple terminals--Part I," \emph{IEEE Trans.\ Inf.\ Theory}, vol.\ 56, no.\ 8, pp.\ 3973--3996, Aug.\ 2010.

\bibitem{Tyagi13}
H.~Tyagi, ``Common information and secret key capacity," \emph{IEEE Trans.\ Inf.\ Theory}, vol.\ 59, no.\ 9, pp.\ 5627--5640, Sep.\ 2013.

\bibitem{Yao79}
A. C.~Yao, ``Some complexity questions related to distributed computing,'' in \emph{Proc.\ 11th Annu.\ ACM Symp.\ Theory of Computing (STOC)}, 1979.

\bibitem{Wyner75} 
A. D.~Wyner, ``The common information of two dependent random variables,'' \emph{IEEE Trans.\ Inf.\ Theory}, vol.\ IT-21, no.\ 2, pp.\ 163--179, Mar.\ 1975.

\bibitem{MKS14}
M.~Mukherjee, N.~Kashyap and Y.~Sankarasubramaniam, ``Achieving SK capacity in the source model: When must all terminals talk?," in \emph{Proc.\ 2014 IEEE Int.\ Symp.\ Inf.\ Theory (ISIT 2014)}, Honolulu, Hawai'i, USA, June 29 -- July 4, 2014, pp. \ 1156--1160.

\bibitem{Chan14}
C.~Chan, A.~Al-Bashabsheh, J.~Ebrahimi, T.~Kaced and T.~Liu, ``Multivariate mutual information inspired by secret key agreement," draft manuscript, Oct.\ 2014 [Online]. 
Available: \url{https://www.dropbox.com/s/q9ru4d0bsyw6per/main.pdf}.

\bibitem{ZLL15}
H.~Zhang, Y.~Liang and L.~Lai, ``Secret key capacity: Talk or keep silent?," in \emph{Proc.\ 2015 IEEE Int.\ Symp.\ Inf.\ Theory (ISIT 2015)}, Hong Kong, China, June 14--19, 2015, pp.\ 291--295.

\bibitem{CW14}
T. A.~Courtade and R. D.~Wesel, ``Coded cooperative data exchange in multihop networks,"  \emph{IEEE Trans.\ Inf.\ Theory}, vol.\ 60, no.\ 2, pp.\ 1136--1158, Feb.\ 2014.

\bibitem{CHISIT14}
T. A.~Courtade and T. R.~Halford, ``Coded cooperative data exchange for a secret key," in \emph{Proc.\ 2014 IEEE Int.\ Symp.\ Inf.\ Theory (ISIT 2014)}, Honolulu, Hawai'i, USA, June 29 -- July 4, 2014, pp.\ 776--780.

\bibitem{CH14}
T. A.~Courtade and T. R.~Halford, ``Coded cooperative data exchange for a secret key," Arxiv:1407.0333v1.

\bibitem{ESS10}
S.\ El Rouayheb, A.\ Sprintson, and P.\ Sadeghi, ``On coding for cooperative data exchange,'' in 
\emph{Proc.\ 2010 IEEE Inf.\ Theory Workshop (ITW 2010)}, Cairo, Egypt, 6--8  Jan.\ 2010, pp.\ 1--5.

\bibitem{LCV15}
J.~Liu, P.~Cuff and S.~Verdu, ``Secret key generation with one communicator and a strong converse via hypercontractivity", in \emph{Proc.\ 2015 IEEE Int.\ Symp.\ Inf.\ Theory (ISIT 2015)}, Hong Kong, China, June 14--19, 2015, pp.\ 710--714.

\bibitem{BR14}
M.~Braverman and A.~Rao, ``Information equals amortized communication," \emph{IEEE Trans.\ Inf.\ Theory}, vol.\ 60, pp.\ 6058--6069, Oct. \ 2014.

\bibitem{BS15}
M.~Braverman and J.~Schneider, ``Information complexity is computable," \emph{Electronic Colloquium on Computational Complexity (ECCC)}, Report No. \ 23, 2015. 

\bibitem{MT10} 
M.~Madiman and P.~Tetali, ``Information inequalities for joint distributions, with interpretations and applications,'' \emph{IEEE Trans.\ Inf.\ Theory}, vol.\ 56, no.\ 6, pp.\ 2699--2713, June 2010.

\bibitem{CZ10}
C.~Chan and L.~Zheng, ``Mutual dependence for secret key agreement," in \emph{Proc.\ 44th Annual Conference on Information Sciences and Systems (CISS)}, 2010.

\bibitem{XLC13}
G.~Xu, W.~Liu and B.~Chen, ``Wyner's common information: Generalizations and a new lossy source coding interpretation," Arxiv:1301.2237v1.

\bibitem{TSP11}
R.~Tandon, L.~Sankar and H.V.~Poor, ``Multi-user privacy: The Gray-Wyner system and generalized common information," in \emph{Proc. \ 2011 IEEE Int.\ Symp.\ Inf.\ Theory (ISIT 2011)}, St.\ Petersburg, Russia, July 31 -- Aug.\ 5, 2011, pp.\ 563--567.

\bibitem{MK14}
M.~Mukherjee and N.~Kashyap, ``On the communication complexity of secret key generation in the multiterminal source model," in \emph{Proc.\ 2014 IEEE Int.\ Symp.\ Inf.\ Theory (ISIT 2014)}, Honolulu, Hawai'i, USA, June 29 -- July 4, 2014, pp.\ 1151--1155.

\bibitem{MK15}
M.~Mukherjee and N.~Kashyap, ``The communication complexity of achieving SK capacity in a class of PIN models," in \emph{Proc.\ 2015 IEEE Int.\ Symp.\ Inf.\ Theory (ISIT2015)}, Hong Kong, China, June 14--19, 2015, pp.\ 296--300.

\bibitem{ElK11}
A.~El Gamal and Y. H.~Kim, \emph{Network Information Theory}, Cambridge University Press, 2011.

\bibitem{Odl95}
A. M.~Odlyzko, ``Asymptotic enumeration methods,'' in \emph{Handbook of Combinatorics}, R.L.~Graham et al., eds., 1995, pp.\ 1063--1229.

\bibitem{TKSV12}
H.~Tyagi, N.~Kashyap, Y.~Sankarasubramaniam and K.~Viswanathan, ``Fault tolerant secret key generation,'' in \emph{Proc.\ 2012 IEEE Int.\ Symp.\ Inf.\ Theory (ISIT 2012)}, Cambridge, Massachusetts, USA, July 1--6, 2012, pp.\ 1787--1791.

\bibitem{Har62}
F.~Harary, ``Maximum connectivity of a graph,'' in \emph{Proc.\ Nat.\ Acad.\ Sci.}, vol.\ 48, pp.\ 1142--1145, 1962.

\bibitem{KMS13}
N.~Kashyap, M.~Mukherjee and Y.~Sankarasubramaniam, ``On the secret key capacity of the Harary graph PIN model," in \emph{Proc.\ 2013 Nat.\ Conf.\ Commun.\ (NCC 2013)}, Delhi, India, Feb.~15--17, 2013, pp.\ 1--5.

\bibitem{CR99}
C. J.~Colbourn and A.~Rosa, \emph{Triple Systems}, Oxford Mathematical Monographs, 1999.

\end{thebibliography}
\end{document}